    \documentclass[lettersize,journal]{IEEEtran}

\usepackage{amsmath,graphics,amssymb,epsfig,subfigure,color,amsthm,cite}
\usepackage{array,booktabs}
\usepackage{multirow}
\usepackage{enumerate}
\usepackage{algorithm}
\usepackage{algorithmicx,algpseudocode}
\usepackage{bm}
\usepackage{float}
\usepackage{graphicx, caption}
\usepackage{adjustbox,rotating}
\usepackage{tabularx} 
\usepackage{makecell} 
\usepackage{graphicx}

\usepackage{verbatim}
\usepackage{color}

\usepackage{xcolor}

\newtheorem{thm}{Theorem}
\newtheorem{lem}{Lemma}

\renewcommand{\thefigure}{\arabic{figure}}

\algnewcommand{\Initialize}[1]{%
  \State \textbf{Initialize:}
  \Statex \hspace*{\algorithmicindent}\parbox[t]{.8\linewidth}{\raggedright #1}
}

\newcommand{\ITO}{It$\hat{\mbox{o}}$}

\newcommand\blfootnote[1]{%
  \begingroup
  \renewcommand\thefootnote{}\footnote{#1}%
  \addtocounter{footnote}{-1}%
  \endgroup
}

\makeatletter
\newcommand{\vast}{\bBigg@{4.5}}
\newcommand{\Vast}{\bBigg@{7.5}}
\makeatother

\begin{document}

\title{  Nonlinear Self-Interference Cancellation with Learnable Orthonormal Polynomials for Full-Duplex Wireless Systems}


\author{Hyowon~Lee,~\IEEEmembership{Graduate Student Member,~IEEE,} Jungyeon~Kim,~\IEEEmembership{Graduate Student Member,~IEEE,}  Geon~Choi,~\IEEEmembership{Graduate Student Member,~IEEE,} Ian P.~Roberts,~\IEEEmembership{Member,~IEEE,} Jinseok Choi,~\IEEEmembership{Member,~IEEE,} \\     and~Namyoon~Lee,~\IEEEmembership{Senior Member,~IEEE}
\thanks{H.~Lee, J.~Kim, and G.~Choi are with the Department of Electrical Engineering, POSTECH, Pohang, South Korea (e-mail: \{hyowon, jungyeon.kim, simon03062\}@postech.ac.kr).}
\thanks{I.~P.~Roberts is with the Department of Electrical and Computer Engineering, UCLA, Los Angeles, CA 90095, USA (e-mail: ianroberts@ucla.edu).} 
\thanks{J.~Choi is with the School of Electrical Engineering, Korea Advanced Institute of Science and Technology (KAIST), Daejeon 34141, South Korea (e-mail: jinseok@kaist.ac.kr).}
\thanks{N.~Lee is with the School of Electrical Engineering, Korea University, Seoul, South Korea (e-mail: namyoon@korea.ac.kr).}}
\vspace{-2mm}	

\maketitle
\vspace{-12mm}

\begin{abstract} 
Nonlinear self-interference cancellation (SIC) is essential for full-duplex communication systems, which can offer twice the spectral efficiency of traditional half-duplex systems. The challenge of nonlinear SIC is similar to the classic problem of system identification in adaptive filter theory, whose crux lies in identifying the optimal nonlinear basis functions for a nonlinear system. This becomes especially difficult when the system input has a non-stationary distribution. In this paper, we propose a novel algorithm for nonlinear digital SIC that adaptively constructs orthonormal polynomial basis functions according to the non-stationary moments of the transmit signal. By combining these basis functions with the least mean squares (LMS) algorithm, we introduce a new SIC technique, called as the adaptive orthonormal polynomial LMS (AOP-LMS) algorithm. To reduce computational complexity for practical systems, we augment our approach with a precomputed look-up table, which maps a given modulation and coding scheme to its corresponding basis functions. Numerical simulation indicates that our proposed method surpasses existing state-of-the-art SIC algorithms in terms of convergence speed and mean squared error when the transmit signal is non-stationary, such as with adaptive modulation and coding. Experimental evaluation with a wireless testbed confirms that our proposed approach outperforms existing digital SIC algorithms.




\end{abstract}


\section{Introduction}\label{sec1}

\nocite{lee2023nonlinear}%
 \blfootnote{The conference version of of this paper was published in IEEE ICC 2023 \cite{lee2023nonlinear}.}

For decades, wireless communication  {systems} have relied on half-duplex operation to decouple uplink and downlink in frequency and/or time to prevent  {manifesting} so-called self-interference (SI).  {Orthogonalizing resources in this way}, however, incurs an inherent loss in spectrum utilization. To reduce this inefficiency in cellular and Wi-Fi networks, there has been growing interest in in-band full-duplex (FD) wireless systems \cite{Jain:2011,duarte2012experiment,Sabharwal:2010,roberts_chapter},  {which} have the potential to double spectral efficiency compared to traditional half-duplex systems \cite{roberts_chapter,duarte2013design,Sabharwal:2014band,Yin:fullduplex_in_wireless}. 
 Beyond this, FD  {technology} {serves as} a crucial  {enabler} of joint communication and sensing systems \cite{amjad2017full_sensing,Barneto2021_sensing} {and of spectrum sharing and cognitive radio \cite{Etkin:spectrum_sharing,Wang:cognitive}.} In recent years, multiple experimental demonstrations have shown that FD operation is indeed possible in real-world systems {\cite{amjad2017low,chung2015prototyping,roberts_realworld}}.  
 
Sophisticated self-interference cancellation (SIC)  techniques are essential to successfully enabling FD wireless systems \cite{roberts_chapter}. 
State-of-the-art SIC typically involves {multiple stages of} SI reduction, including i) passive SI reduction using circulators or other RF isolation and ii) active cancellation involving analog SIC and digital nonlinear SIC \cite{everett2014passive,Ahmed2015all,duarte2012experiment,Sabharwal:2010,ahmed2013self_nonlinear_distortion,Bharadia2013full,Mikko2015_analog_digital,JY2023Onlearning_magazine}. 
It is well known that analog SIC is {often} necessary to prevent SI from saturating analog-to-digital converters (ADCs) {and other receive chain components} \cite{Sabharwal:2014band,roberts_chapter}. 
Afterwards, digital SIC plays a key role in canceling the residual SI which remains after analog SIC. 
Eliminating this residual SI to the noise  {floor} digitally, however, has proven to be extremely challenging \cite{Korpi2014widely,Sabharwal:2014band}, largely due to (i) the nonlinearity introduced by power amplifiers (PAs), IQ imbalance, and phase-noise, and (ii) the time-varying nature of SI \cite{ahmed2013self_nonlinear_distortion,Korpi2014widely,korpi2016full_timevarying}. 

The time-varying and nonlinear SIC problem is mathematically equivalent to a classical time-varying and nonlinear system identification problem in adaptive filter theory \cite{Haykin_Book2008}. 
The most popular approach to approximate such a system is using the Wiener-Hammerstein model comprised of parallel Hammerstein polynomials (HP) cascaded with linear finite impulse response (FIR) filters \cite{schoukens2011parametric_HPmodel}. 
Using this approximation model, the most straightforward online SIC algorithm, called the HP-LMS algorithm, was proposed in \cite{Book:ding2004digital}. 
The  {key} idea behind this approach is to adopt the least mean squares (LMS) algorithm to optimize the FIR filter coefficients  {using HP basis functions to accurately capture nonlinearity.} 
This SIC algorithm is simple and can indeed eliminate SI to {near} the noise floor, provided that the orders of the HP are  {sufficiently high}. 
The primary drawback of this algorithm is that it suffers from slow convergence due to  {correlations between {HP} basis functions of different orders}.


To boost the convergence speed in  {adapting} the FIR filter coefficients, \cite{Korpi:2015} introduces an orthogonal transformation with a whitening filter to eliminate cross-correlation across nonlinear basis functions. 
We refer to this as an HP-based recursive least squares (RLS) (HP-RLS) algorithm \cite{korpi2016full_timevarying,JY2023Onlearning_magazine}. 
This orthogonal transformation process, however,  {involves} high computational complexity in estimating the sample covariance matrix and in computing the inverse of the covariance matrix for the whitening filter. 
To reduce  {this} complexity, the work  {of} \cite{Jungyeon:2018} proposed harnessing a set of orthogonal basis functions called \ITO-Hermite (IH) polynomials for SIC, assuming the transmit data follows  {the} complex Gaussian distribution. 
In such cases, the IH-LMS algorithm has shown to achieve  {a} convergence speed  {on par with that of} the HP-RLS algorithm, even without estimating the sample covariance matrix \cite{Jungyeon:2018}. 
Although the IH-LMS algorithm can significantly reduce the computational complexity while achieving fast convergence,  {it relies on}  {a} complex Gaussian input. 
This assumption can indeed often be valid in wireless systems which employ orthogonal frequency-division multiplexing (OFDM), {since, with a large number of subcarriers,} the transmitted OFDM symbols  {are} approximately distributed as complex Gaussian by the central limit theorem \cite{araujo2011accuracy}. 
When the transmitted signal distribution is {not} complex Gaussian, however, IH-LMS can suffer from slow convergence  {due} to correlations among basis functions \cite{Haykin_Book2008}.


{In addition to such assumptions on the signal distribution,} most prior studies {on digital SIC} have focused on the case of a stationary transmit signal,  {where} its statistical distribution does not change over time. 
In practice, however,  {this stationary assumption} is often not appropriate, as  {real-world wireless systems employ} adaptive modulation and coding  {and transmit power control\cite{rashid1998transmit,lee2014power}}, which naturally leads to changes in the distribution of the transmitted signal. 
In 5G, for instance, such link adaptation can happen across a single mini-slot 
 \cite{Dulek:2017_ModulationClassification}. 
 {This motivates the need} to design a SIC algorithm  {which accommodates} non-stationary transmit signals in order to  {successfully} realize  {FD operation} in 5G and beyond. 
 {We develop such an algorithm herein, which, to the best of our knowledge, is the first of its kind.}
\subsection{Contributions}
The contributions of this paper are summarized as follows:
\begin{itemize}
    \item We first show the existence of a set of orthonormal nonlinear basis functions which relate the transmit signal to received SI.
    The key idea to constructing the proposed basis functions is to generalize the IH polynomials according to the (possibly time-varying) moments of the transmitted signal. 
    \item  Then, we propose a computationally efficient algorithm that constructs the orthonormal nonlinear basis functions. 
    By harnessing the recursive structure in identifying the coefficients of IH polynomials, we show that 
     {the construction of $p$th-order orthonormal polynomial is possible with a computational complexity of $\mathcal{O}(p^2)$.} 




    
\item  Cascading the constructed  {orthonormal} basis functions with the LMS algorithm, we present a SIC algorithm for non-stationary input data,  {which we refer to} as the AOP-LMS algorithm. 
In the moment learning phase of this algorithm, the moments of transmitted data symbols are estimated and then the set of  {orthonormal} basis functions is systematically constructed. 
 {Then,} in the filter coefficient learning phase,  {our algorithm uses the LMS algorithm to refine and adapt these coefficients.} 

\item  {We further accelerate the proposed approach with a look-up table containing precomputed orthonormal basis functions for various signal distributions (e.g, $M$-QAM). In practice, such a look-up table can be constructed a priori based on the known set of modulation and coding schemes (MCSs) employed by a transceiver and can then be referenced in real  time as the MCS changes.}


\item  Using numerical simulation, we show that our proposed SIC algorithm provides significant  {improvement} in both mean squared error (MSE) and convergence speed compared to the state-of-the-art algorithms, including the HP-RLS algorithm, under non-stationary  {input distributions}. 

\item  We also verify the effectiveness of the proposed SIC algorithm using a real-time wireless testbed.  {These experimental results further confirm that our proposed approach---and its ability to adapt to changes in modulation---translates from theory to implementation.}

\end{itemize}

\subsection{Organization}

This paper is  {organized} as follows. Section II defines the system model of the FD wireless communication system and  {motivates} the need for pair-wise orthogonality with respect to the basis of nonlinear LMS algorithms. Section III  {presents a method} to construct  {orthonormal} polynomials using partial moments of the input signal. 
Section IV provides orthonormal polynomials  {for} various input signal distributions.  {Section V introduces the proposed digital SIC algorithms based on the formulations laid forth in Sections III-IV.} Section VI and Section VII show results of  {the proposed} SIC algorithms  {through numerical simulation and real-world experiments.} Section VIII concludes the paper.

\section{System Model}\label{Sec:system_model}



We consider a typical FD transceiver compromised employing digital SIC, as illustrated in Fig.~\ref{fig:system_model}.
We do not explicitly consider analog SIC in this work, but rather assume it could be applied directly in conjunction with the method proposed herein.
We denote the complex baseband transmit  {sample} at time slot $n$ by $x[n]$. 
The  {samples} of $\{x[n]\}$  {are} assumed to  {be statistically} independent and identically distributed (IID) and  {to be non-stationary.}  
The complex baseband symbol, $x[n]$, passes through a digital-to-analog converter (DAC), upconversion, a nonlinear PA, the SI channel, a low-noise amplifier (LNA), downconversion, and an ADC before being observed digitally at the receiver.
The collective response of these components and processes  is in general  a nonlinear system.
It is often assumed that the dominant nonlinear effects are introduced by high-order harmonics of practical PAs \cite{korpi2014full_system_calculation,Korpi2014widely}, but the proposed method herein does not rely on this assumption. 
Let $y[n]$ be the received signal at time slot $n$, which is assumed to be some nonlinear  {combination of the current and} previous $L-1$ transmit  {samples} ${\bf x}_L[n] \triangleq \left[x[n],x[n-1],\ldots, x[n-L+1]\right]^{\top}$. 
 {To capture this, let us define a nonlinear function $f:\mathbb{C}^L\rightarrow \mathbb{C}$ which relates the received signal at time slot $n$ to the transmitted signal as}
   \begin{align}
 	y[n]= f({\bf x}_L[n]) +z[n], \label{eq:nonlinear_sys}
 \end{align}
 where $z[n]$ is IID complex Gaussian  {noise} with zero mean and variance $\sigma^2$.  {It is assumed that} this nonlinear function $f(\cdot)$ is unknown to the transceiver and  {that it may change over time.} 
 
 \begin{figure}
	\centering
	\includegraphics[width=0.35\textwidth]{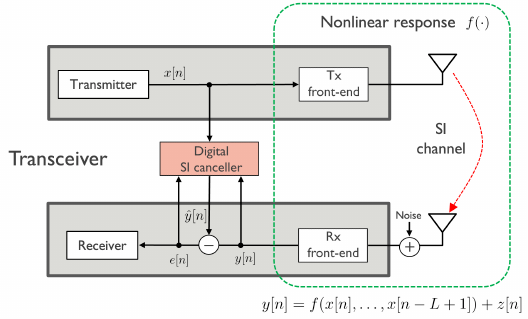}
    \caption{A  {full-duplex} transceiver that  {employs digital SIC to cancel} both linear and nonlinear components of SI.} \label{fig:system_model}
\end{figure}

\subsection{Model-Based Function Approximation}

The main task of digital  {SIC} is to find the best approximation of the nonlinear function $f(\cdot)$ in an online manner---a problem akin to those appearing in contexts of adaptive filter theory as system identification \cite{Haykin_Book2008}. 
 {To} develop a model for this nonlinear function approximation,  {we harness established models for} nonlinear PAs and linear filters. 
    
We begin by considering Saleh's PA model, a time-honored wideband PA model with memory effects initially introduced in \cite{saleh1981frequency,al2020simplified}. 
 {Under such a model, the output of a PA with memory} $M$,  {whose input is} $x[n]$,  {is modeled as} 
    \begin{align}
     x_{\sf PA}[n]= \sum_{m=0}^{M-1} h^{\sf PA}[m] \frac{\gamma x[n-m]}{1+\beta |x[n-m]|^2}, \label{eq:Saleh_model}
\end{align}
where $\beta \geq 0$ and $\gamma\geq 0$ capture the transition sharpness and small signal gain of the PA, respectively. 
Here, $h^{\sf PA}[m]$ is the $m$th coefficient of the PA impulse response. 
We assume that $\beta, \gamma,$ and $\{h^{\sf PA}[m]\}$ are initially unknown to the system  {and may vary with time (e.g., due to temperature)}. 

Comparatively, Saleh's model introduces more significant AM/PM distortion
than most  {off-the-shelf} solid-state PAs \cite{saleh1981frequency,rapp:1991effects,RF_SSPA}. 
 {Consequently, employing this model will provide flexibility and robustness in capturing (potentially other) sources of nonlinearity beyond solely AM/AM distortion.}

Using the truncated power series expansion with maximum  degree $P$ (odd integer), we obtain an approximation of the nonlinear PA transfer function \eqref{eq:Saleh_model} in a form of the linear combination of $\frac{P+1}{2}$ HPs 
 $|x[n]|^{2p}x[n] $ for $p \in \left\{0,1,\ldots, \frac{P-1}{2} \right\}$:
\begin{align}
	\frac{\gamma x[n]}{1+\beta |x[n]|^2} = \sum_{p=0}^{\frac{P-1}{2}}(-1)^p \gamma \beta^p |x[n]|^{2p}x[n] +\mathcal{O}\left((x[n])^{P+1}\right). \label{eq:Saleh_model_approx}
\end{align}
Without loss of generality, we generalize the $p$th order nonlinear basis function as a linear combination of the HPs with maximum degree $p$ (odd integer), namely
 \begin{align}
	\phi_p(x[n]; {\bf c}_p) = \sum_{k=0}^{p-1} c_{p,k}  |x[n]|^{2k}x[n],   \label{eq:GIH}
\end{align}
where $c_{p,k}\in \mathbb{C}$ is the $k$th coefficient of the $p$th basis function and ${\bf c}_p=[c_{p,0}, \ldots, c_{p,p-1}]$. 
This overparameterized representation of the basis function provides additional degrees of freedom to capture PA nonlinearity. 
For instance, when $c_{p,k}=0$ for $k\in \left\{0,1,\ldots, p-2\right\}$ and $c_{p,p-1}=1$, we obtain the standard  {HPs} as the basis functions, i.e., $\phi_p^{\sf HP}(x[n]) = |x[n]|^{2(p-1)}x[n]$.   
By tuning $c_{p,k}$,  different classes of HPs can be generated as basis functions. 
Ignoring the higher-order approximation error in \eqref{eq:Saleh_model_approx}, we can express the PA function in \eqref{eq:Saleh_model_approx} by the sum of the $\frac{P+1}{2}$ overparameterized nonlinear basis functions $\phi_p(x[n])$ as
\begin{align}
	\frac{\gamma x[n]}{1+\beta |x[n]|^2} = \sum_{p=1 }^{\frac{P+1}{2}} \phi_p(x[n]; {\bf c}_p) {.} \label{eq:Saleh_model_approx2}
\end{align}
Plugging \eqref{eq:Saleh_model_approx2} into \eqref{eq:Saleh_model}, we can rewrite the PA output using the parallel HP with memory length $M$ as
 \begin{align}
    x_{\sf PA}[n]&=  \sum_{m=0}^{M-1} h^{\sf PA}[m] \left( \sum_{p=1}^{\frac{P+1}{2}} \phi_p(x[n-m]; {\bf c}_p) \right)  \\
    &= \sum_{p=1 }^{\frac{P+1}{2}} \sum_{m=0}^{M-1} h_{p}^{\sf PA}[m] \phi_p(x[n-m]; {\bf c}_p),\label{eq:PA_out}
\end{align}
where $h_{p}^{\sf PA}[m]$ is the impulse response of the PA for the $p$th order nonlinear input $\phi_p(x[n])$.  
The latter equality comes from  {our} overparameterization  {approach, which will allow us to} capture the input-dependent PA memory response effects $h_{p}^{\sf PA}[m]$ for $p\in \{1, 2, \dots, \frac{P+1}{2}\}$.

Let $\{h_{\sf SI}[q]\}$ be the length-$Q$ impulse response of the SI  {propagation} channel. 
Under  {an idealized} LNA and ADC\footnote{PAs are often assumed the dominant sources of nonlinear SI~\cite{korpi2014full_system_calculation,Korpi2014widely}; our model readily accommodates other sources of nonlinearity, nonetheless.}, the received baseband SI signal is approximately
\begin{align}
	{\hat y}[n]=  \sum_{q =0}^{Q-1} h_{\sf SI}[q] x_{\sf PA}[n-q]. \label{eq:received}
\end{align} 
Invoking \eqref{eq:PA_out} into \eqref{eq:received}, this approximated SI is equivalently
\begin{align}
	{\hat y}[n] &=  \sum_{p=1 }^{\frac{P+1}{2}} \sum_{\ell=0}^{L-1} h_{p}[\ell] \phi_p(x[n-\ell]; {\bf c}_p) \nonumber\\
	&=\underbrace{\sum_{p=1 }^{\frac{P+1}{2}}   \sum_{\ell=0}^{L-1} h_{p}[\ell]  \left(\sum_{k=0}^{p-1} c_{p,k}  |x[n-\ell]|^{2k}x[n-\ell]\right)}_{={\hat f}({\bf x}_L[n] ;\Theta_1,\Theta_2)}, \label{eq:received2}
\end{align} 	
where $\{h_{p}[\ell]\}$ is the FIR filter containing the combined effects of  $h_{p}^{\sf PA}[m]$ and $ h_{\sf SI}[q]$, with  {total memory} $L=M+Q-1$. 
Our approximation of the SI signal in \eqref{eq:received2} contains two  {sets of unknown parameters}: 
\begin{enumerate}
    \item The coefficients for constructing the nonlinear basis functions $\Theta_1=\left\{{\bf c}_{1}, \ldots, {\bf c}_{\frac{P+1}{2}}\right\}$, with $|\Theta_1|=\frac{(P+1)(P+3)}{8}$.
    \item The coefficients of the effective FIR filter,  {$\Theta_2=\{ 
    \{h_{1}[\ell]\},\{h_{2}[\ell]\},\dots,\{h_{\frac{P+1}{2}}[\ell]\}\}$}, with $|\Theta_2|=\frac{(P+1)L}{2}$. 
\end{enumerate}
We denote our approximation of the  {effective nonlinear SI channel} in \eqref{eq:nonlinear_sys} by ${\hat f}(\cdot ;\Theta_1,\Theta_2):\mathbb{C}^L\rightarrow \mathbb{C}$ with parameters $\Theta_1$ and $\Theta_2$.

We  {emphasize} that this model-based function approximation method differs from the  {model-free} function approximation techniques with deep neural networks (DNNs) \cite{Alexios2018neural,Kong2022Neural},  {which} approximate SI without exploiting any  {prior} model knowledge.  
 {Our} model-based method,  {on the other hand, leverages} select parameters  based on knowledge of established PA models and the FIR filters of linear time-invariant systems. 
Ultimately, this proves both numerically and experimentally to make our proposed approach more adaptive and effective in canceling SI, as we will see in Sections VI and VII.

\subsection{Problem Statement}

Let $e[n]$ be the error between the  {true} received SI $y[n]$ and  {the SI} constructed by our approximation with input ${\bf x}_L[n]$:
\begin{align} 
	e[n]& = y[n] - {\hat y}[n] \nonumber\\
	        &= { f}({\bf x}_L[n]) - {\hat f}({\bf x}_L[n] ;\Theta_1,\Theta_2) +w[n].
\end{align}
 {With this function approximation in hand,} our goal is to find  {the parameter sets $\Theta_1$ and $\Theta_2$ which minimize the empirical squared error}
\begin{align}
	J(\Theta_1,\Theta_2) = \sum_{n \in \mathcal{T}}|e[n]|^2,
\end{align}
where $\mathcal{T}=\{n_1,n_1+1,\ldots, n_2\}\subset \mathbb{Z}$ is an index set capturing the time interval of interest. 
In general, finding a jointly optimal solution of $\Theta_1$ and $\Theta_2$ is very challenging due to the non-convexity of  {$J(\Theta_1,\Theta_2)$} with respect to both parameter sets.  
 {In light of this difficulty, we instead optimize  {them} in a disjoint manner using classical adaptive filter theory.}

Let ${\bm \phi}_p({\bf x}_L[n];{\bf c}_p) =\left[\phi_p(x[n];{\bf c}_p), \phi_p(x[n-1]; {\bf c}_p), \ldots,  \right.$ $\left. \phi_p(x[n-L+1]; {\bf c}_p) \right]$ be the filter input vector generated by the $p$th order basis function and ${\bf h}_p[n]=\left[h_p[n], h_p[n-1],\ldots, h_p[n-L+1] \right] $ be the filter response vector of the $p$th order basis function. By concatenating the input and filter response vectors, we can also define 
\begin{align}
{\bm \phi}({\bf x}_L[n]; \Theta_1) =&\left[{\bm \phi}_1({\bf x}_L[n];{\bf c}_p),\ldots, \right. \nonumber \\ &\left. {\bm \phi}_{\frac{P+1}{2}}\left({\bf x}_L[n];{\bf c}_{\frac{P+1}{2}}\right) \right]^{\sf H} \in \mathbb{C}^{\frac{L(P+1)}{2} \times 1}
\end{align}
and ${\bf h}[n]=\left[{\bf h}_1[n],\ldots, {\bf h}_{\frac{P+1}{2}}[n] \right]^{\top} \in \mathbb{C}^{\frac{L(P+1)}{2} \times 1}$.  Then, for given $\Theta_1=\left\{{\bf c}_1,\ldots, {\bf c}_{\frac{P+1}{2}}\right\}$, the approximation of SI in \eqref{eq:received2} can be rewritten in a linear form with respect to ${\bf h}[n]$ as
\begin{align}
	{\hat y}[n] 		&={\hat f}({\bf x}_L[n] ;\Theta_1,\Theta_2) \nonumber\\
	&=  {\bm \phi}({\bf x}_L[n]; \Theta_1)^{\sf H}{\bf h}[n].
\end{align}
 {In the next section, we delineate the methodology for optimizing the basis parameter $\Theta_1$ of $\boldsymbol{\phi}$, leveraging the characteristics of the transmitted signals $\{\mathbf{x}[n]\}$.}
\section{Orthonormal Polynomial Construction}

 {In this section,} we first introduce a method for selecting coefficients that satisfy the orthonormal condition. 
We then propose a low-complexity algorithm that finds such coefficients using Schur complement inversion. 


\subsection{Orthonormal Polynomial Construction using Moments}

Using the moments of  {the transmit} signal $x[n]$, our goal is to construct a set of orthonormal basis functions $\phi_1(x[n];\mathbf{c}_1), \phi_2(x[n];\mathbf{c}_2),\ldots, \phi_{\frac{P+1}{2}}(x[n];\mathbf{c}_{\frac{P+1}{2}})$ that  {satisfy} the following conditions:
\begin{align}\label{eq:orthogonal_condition}
    \begin{cases}
    \mathbb{E}[\phi^*_i(x[n];\mathbf{c}_i)\phi_j(x[n];\mathbf{c}_j)] = 1, & \mbox{if } i=j, \\
    \mathbb{E}[\phi^*_i(x[n];\mathbf{c}_i)\phi_j(x[n];\mathbf{c}_j)] = 0, & \mbox{if } i\neq j.
    \end{cases}
\end{align}
Before  {proceeding}, we first define a moment vector $\boldsymbol{\mu}_{2a}^{2b}$ and a Hankel matrix $\mathbf{M}_p$ as
\begin{align}
    \pmb{\mu}_{2a}^{2b} \triangleq \left[\mu_{2a},\mu_{2(a+1)},\dots,\mu_{2b}   \right]^\top,
\end{align}
and
\begin{align} \label{eq:Hankel_mtx}
    \mathbf{M}_p \triangleq  
    \begin{bmatrix} 
        \mu_{2} & \mu_{4} & \dots & \mu_{2p-2} \\
        \mu_{4} & \mu_{6} & \dots & \mu_{2p} \\
        \vdots & \vdots &  \ddots & \vdots \\
        \mu_{2p-2} & \mu_{2p} & \dots & \mu_{4p-6}
    \end{bmatrix},
\end{align}
where $a \leq b$, $p> 1$, and $\mu_{2p} = \mathbb{E}[|x[n]|^{2p}]$. 
Now, we provide a theorem for finding  {polynomial coefficeints} which satisfy the orthonormal condition \eqref{eq:orthogonal_condition}, given known moments $\pmb{\mu}_{2a}^{2b}$.
\begin{thm} \label{thm:orthogonal_coefficient}
     {Given the even moments} $\{ \mu_2, \mu_4, \dots, \mu_{4p-2} \}$,  {the} basis functions $ \{\phi_1(x[n];\mathbf{c}_1), \dots,$ $ \phi_{p} (x[n]; \mathbf{c}_{p}) \}$ are orthonormal provided that
   \begin{align}
       \hat{\mathbf{c}}_p = \frac{1}{z}[\bar{\mathbf{c}}_p^\top, 1]^\top,
   \end{align}
    where 
    \begin{align}\label{eq:equation_solution}
     \bar{\mathbf{c}}_p = -\mathbf{M}_p^{-1} \pmb{\mu}_{2p}^{4p-4},
 \end{align}
 and $z$ is a normalization factor given by
  \begin{align} \label{eq:expect_normal_factor}
    z = \sqrt{\mathbb{E}[|\phi_p(x[n];\mathbf{c}_p)|^2]}.    
 \end{align}
\end{thm}

\begin{proof}
See Appendix 1.
\end{proof}


\begin{figure*}[h]
    \centering
    \includegraphics[scale=0.6]{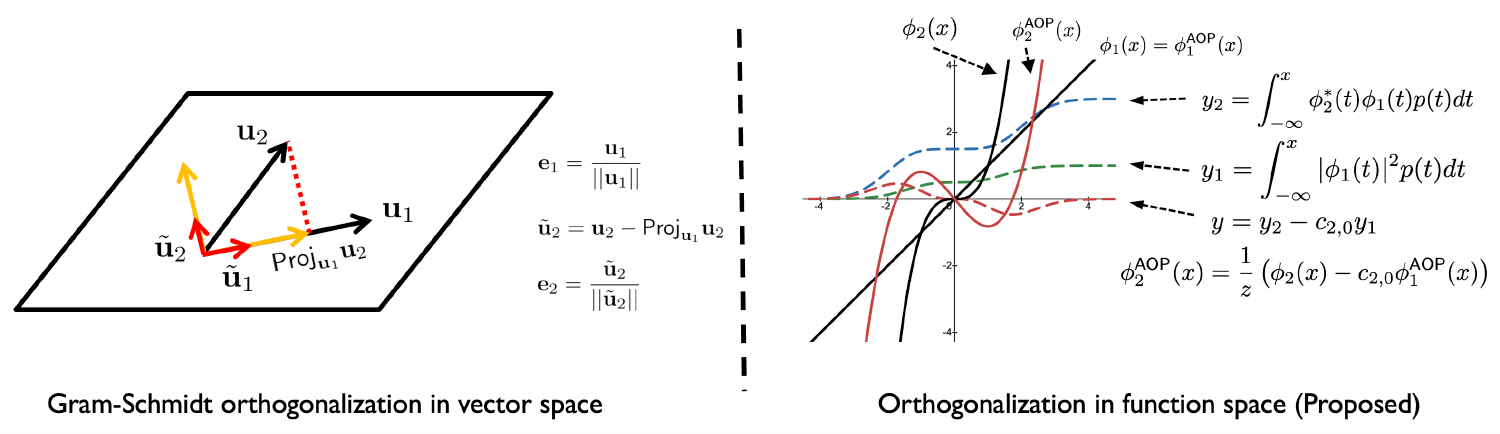}
    \caption{An illustration of the proposed  {orthonormal} polynomial construction method. Similar to  {the} Gram-Schmidt  {orthogonalization} process in vector space, our proposed algorithm sequentially constructs orthonormal polynomials in function space.} 
    \label{fig:gram_schmidt_comparison}
\end{figure*}
From Theorem \ref{thm:orthogonal_coefficient}, we can construct the coefficients of a polynomial satisfying \eqref{eq:orthogonal_condition} when matrix $\mathbf{M}_p$ is nonsingular.  {The process of} sequentially finding the coefficients of orthogonal polynomials using the matrix equation is similar to the orthogonalization algorithm  in  {traditional} linear algebra. 
Extending the  {notion of a} projection in a vector space, this method orthonormalizes in a \textit{function} space, where the inner product is defined as in Fig. \ref{fig:gram_schmidt_comparison}.

The advantage of this method is that it can find the orthonormal polynomials  {with only knowledge of the moments.} 
 {However, one drawback arises from the need  to perform several successive matrix inversions for higher-order orthonormal polynomials.} 
In the following subsection,  {we propose an equivalent yet more efficient method which reduces the complexity of  {such} successive matrix inversions.}

 \begin{algorithm}[h]
 \caption{Adpative orthonormal polynomial construction}\label{alg:orthonormal_polynomial}
 \begin{algorithmic}[1]
  \State \textbf{Input:} $\pmb{\mu}_{2}^{2P}$
  \State \textbf{Output: } $\phi^{\sf AOP}_p(x;\hat{\mathbf{c}}_p),~p=1,\dots,\frac{P+1}{2}$
  \State \textbf{Initialize: }$\mathbf{M}^{-1}_2 \leftarrow \left[\frac{1}\mu_2\right]$,~  $\mathbf{c}_1 \leftarrow [\frac{1}{\sqrt{\mu_2}}] $  
 \For{$p=2$ to $\frac{P+1}{2}$}
         \State $\bar{\mathbf{c}}_{p} \leftarrow -\mathbf{M}_{p}^{-1}\pmb{\mu}_{2p}^{4p-4}$, $\mathbf{c}_{p} \leftarrow [\bar{\mathbf{c}}_{p}, 1] $ \label{state:c}
      \State $\tilde{\mathbf{c}}_{p} \leftarrow \mathbf{c}_{p}*\mathbf{c}_{p}$ \label{State:conv}
      \State $\hat{\mathbf{c}}_{p} \leftarrow \mathbf{c}_{p} / \sqrt{ \sum_{k=1}^{2p-1} \tilde{c}_{p,k} \mu_{2k}} $ \label{State:normal_coeff}
   \State $\mathbf{u}_{p} \leftarrow \pmb{\mu}_{2p}^{4p-4}$ \label{state:U} 
      \State $\tilde{\mathbf{u}}_p \leftarrow \mathbf{M}_p^{-1}\mathbf{u}_{p}$
    \State  $s_{p} \leftarrow \mu_{4p-2}- \mathbf{u}_p^\top\tilde{\mathbf{u}}_p$\label{state:s} 
           
    \State $\mathbf{M}_{p+1}^{-1} \leftarrow  \begin{bmatrix}
      \mathbf{M}_{p}^{-1}+\tilde{\mathbf{u}}_p s_{p}^{-1}\tilde{\mathbf{u}}_p^\top & -\tilde{\mathbf{u}}_p s_{p}^{-1} \\ -s^{-1}_{p}\tilde{\mathbf{u}}_p^\top & s^{-1}_{p}
     \end{bmatrix}$ \label{state:mtx_inv}
 \EndFor
   \end{algorithmic}
\end{algorithm}

\subsection{ {Efficiently} Constructing Orthonormal Polynomials}



 {In light of the computational costs of successive matrix inversions, we now propose an} orthonormal polynomial construction method  {involving} two steps: i)  recursive computation of the orthogonal basis functions and ii) basis function normalization.

{\bf Recursive Computation:}
From Theorem \ref{thm:orthogonal_coefficient}, we construct the matrix equation satisfying the orthogonal condition in \eqref{eq:orthogonal_condition}.
To obtain the  {basis function coefficients}, we need to solve the matrix equations in \eqref{eq:equation_solution} for each $p\in\{1,2,\ldots, \frac{P+1}{2}\}$. 
The computational complexity for this operation scales considerably as the maximum nonlinear order $\frac{P+1}{2}$ increases. 
To reduce the computational complexity in constructing the orthonormal basis functions, we present a recursive algorithm using the structure of the Hankel matrix ${\bf M}_p$. 
 {In essence, we will use $\mathbf{M}_p$ to} construct ${\bf M}_{p+1}$. 
Let $\mathbf{u}_{p}\triangleq\pmb{\mu}_{2p}^{4p-4}$ 
and define a scalar $ r_{p} \triangleq \mu_{4p-2}$. 
Then, $\mathbf{M}_{p+1}$  {can be} constructed as
\begin{align}
     \mathbf{M}_{p+1} =
     \begin{bmatrix}
      \mathbf{M}_{p} & \mathbf{u}_{p} \\ \mathbf{u}_{p}^\top & r_{p}
     \end{bmatrix}.
\end{align}
Using the fact that ${\bf M}_p$ is the Schur complement \cite{golub2013matrix} of $\mathbf{M}_{p+1}$, the inverse of $\mathbf{M}_{p+1}$ can be recursively computed using the inverse of  $\mathbf{M}_{p}$ as
 \begin{align}\label{eq:SchurComplement}
     \mathbf{M}_{p+1}^{-1} &=
     \begin{bmatrix}
       \mathbf{M}_{p}^{-1}+\mathbf{M}_{p}^{-1}\mathbf{u}_{p} s_{p}^{-1}\mathbf{u}_{p}^{\top}\mathbf{M}_{p}^{-1} & -\mathbf{M}_{p}^{-1}\mathbf{u}_{p} s_{p}^{-1} \\ -s^{-1}_{p}\mathbf{u}_{p}^{\top}\mathbf{M}_{p}^{-1} & s^{-1}_{p}
     \end{bmatrix} \nonumber \\
     &=     \begin{bmatrix}
      \mathbf{M}_{p}^{-1}+\tilde{\mathbf{u}}_p s_{p}^{-1}\tilde{\mathbf{u}}_p^\top & -\tilde{\mathbf{u}}_p s_{p}^{-1} \\ -s^{-1}_{p}\tilde{\mathbf{u}}_p^\top & s^{-1}_{p}
     \end{bmatrix},
 \end{align}
 where $\tilde{\mathbf{u}}= \mathbf{M}_p^{-1}\mathbf{u}_p$ and $s_{p} = r_{p}-\mathbf{u}_{p}^{\top}\tilde{\mathbf{u}}_p$. By harnessing the structure of the Hankel matrix ${\bf M}_p$, we can compute a set of the orthogonal basis functions in  {this} computationally efficient manner. This recursive computation  {of $\{\mathbf{M}_p^{-1}\}$}  is summarized in Algorithm \ref{alg:orthonormal_polynomial}.

 {\bf Normalization:} Once the orthogonal basis functions are computed by  {the proposed} recursive method, 
 it is necessary to normalize the polynomial basis function to have unit norm. With this normalization, we obtain the orthonormal basis functions $\phi^{\sf AOP}_{p}(x;\hat{\mathbf{c}}_{p})$  as
\begin{align}
    \phi^{\sf AOP}_{p}(x;\hat{\mathbf{c}}_{p}) = \frac{\phi_{p}(x;\mathbf{c}_{p})}{z},
\end{align}
where $\hat{\mathbf{c}}_{p}$ is an orthonormal coefficient  {vector} of the $p$th order polynomial.
The normalization coefficient $z$ in \eqref{eq:expect_normal_factor} is computed by
\begin{align}\label{eq:normalization_factor}
     z = \sqrt{\mathbb{E}[|\phi_{p}(x;\mathbf{c}_{p})|^2]} =\sum_{k=1}^{2p-1} \tilde{c}_{p,k} \mu_{2k},
\end{align}
where  {$\tilde{c}_{p,k} = \sum_{i} [\mathbf{c}_{p}]_i[\mathbf{c}_{p}]_{k-i}$ is the self-convolution of $\mathbf{c}_p$.}

 {Calculating coefficients for the $p$th orthonormal polynomial using conventional techniques such as Gaussian elimination involves a computational complexity of $\mathcal{O}(\frac{2}{3}p^3)$. 
 {Employing our proposed method reduces this complexity to $\mathcal{O}(3p^2)$.}}

\section{Orthonormal Polynomials for\\ {Various} Signal Distributions}

Using the algorithm outlined in the previous section, this section presents  {orthonormal} polynomials for several representative signal distributions widely used  {in communication systems.} 
Note that some of the distributions considered herein have orthonormal polynomials which have also been investigated in other contexts in prior work. 
We first summarize the moments and  {orthonormal} polynomial coefficients of signals following continuous probability distributions such as Gaussian, uniform, and exponential.
Then, we present  {orthonormal} polynomial coefficients for signals following discrete probability distributions  {commonly used} in  {today's} wireless  {systems}.

\subsection{Orthonormal Polynomials for Complex Gaussian Signals}

 {Let $X$ be a complex Gaussian distributed random variable with mean zero and variance $\sigma_x^2$, denoted as $\mathcal{CN}(0,\sigma_x^2)$, 
which is perhaps the most widely used distribution in communication theory. 
Recall, the only information necessitated by our algorithm to construct an orthonormal polynomial are the moments of the distribution.}
The $m$th-order moment, $\mathbb{E}[|X|^m]$, has a closed-form expression for even $m$ \cite{fassino2019computing:Gaussian_moment},  {given by}
\begin{align}\label{eq:Gaussian_moment}
    \mathbb{E}[|X|^m] = \sigma_x^m \left(\frac{m}{2}\right)!. 
 \end{align}
  For the particular case of $\sigma_x^2 = 1$, the first three  {orthonormal} polynomials obtained by  {our proposed} algorithm are  {then}
   \begin{align}
          \phi_1(X) &= X, \nonumber \\ 
          \phi_2(X) &= \frac{1}{\sqrt{2}}(|X|^2X-2X), \nonumber \\
          \phi_3(X) &= \frac{1}{\sqrt{12}}(|X|^4X-6|X|^2X+6X). 
     \end{align}
 {We point out that the orthonormal polynomials of complex Gaussian random variables have been studied  {in prior literature under the moniker of} \ITO-Hermite polynomials \cite{ito1952complex}.} 

\subsection{Orthonormal Polynomials for Uniform Signals}
 {The random variable $X$ uniformly distributed over the interval $[-k,k]$, where $k>0$}, has a probability density function (PDF) given by
\begin{align}
f_X(x) = 
\begin{cases} 
\frac{1}{2k} & \text{if } -k \leq x \leq k \\
0 & \text{otherwise}.
\end{cases}
\end{align}
It can be shown  {straighforwardly} that $\mathbb{E}[X^{2m}]=\frac{1}{2m+1}k^{2m}$. 
 {Since} the matrix $\mathbf{M}_p$ in \eqref{eq:Hankel_mtx} is non-singular  {when populated with} moments of the uniform distribution,  {the corresponding polynomial coefficients can be obtained directly from Algorithm~\ref{alg:orthonormal_polynomial}}. 
For the particular case of $k = 1$, the first three  {orthonormal} polynomials obtained by  {Algorithm~\ref{alg:orthonormal_polynomial}} are:
   \begin{align}\label{eq:uniform_op}
          \phi_1(X) &= \sqrt{3}X, \nonumber \\ 
          \phi_2(X) &= \sqrt{\frac{7}{4}}(5|X|^2X-3X), \nonumber \\
          \phi_3(X) &= \sqrt{\frac{11}{64}}(63|X|^4X-70|X|^2X+15X).
     \end{align}
Note that it is  {well} known that the  {orthonormal} polynomials of the uniform distribution are the Legendre polynomials \cite{weisstein2002legendre}, with \eqref{eq:uniform_op} being a scaled version of such. 

\subsection{Orthonormal Polynomials for Exponential Signals}
 {Suppose} $X$  {is an exponential random variable} parameterized by $\lambda>0$,  {having} PDF
\begin{align}
    f_X(x) = \lambda \exp(-\lambda x),  {\quad x > 0}.
\end{align}
The moment generating function (MGF) is  {well} known  {to be} $M_X(t) = \frac{\lambda}{\lambda-t}$. 
Moments of this distribution  {can be} derived from the MGF as $\mathbb{E}[X^m] = \frac{\partial^m M_X(t)}{\partial t^m}|_{t=0}$ . 

 {Note that distributions with non-zero mean, such as this, may yield multiple valid orthonormal polynomials, depending on if one considers only even degrees or both even and odd degrees. In the former case, the first polynomial is denoted as $\phi_1(x) = x$, whereas in the latter case, the degree of the first polynomial is $0$, leading it to be denoted as $\Phi_0 = c_{0,0}$, where $\Phi_p$ is extended basis function defined as}
\begin{align}
    \Phi_p(X;\mathbf{c}^p) = \sum_{k=0}^{p} c_{p,k}|X|^{k}X.
\end{align}
The proposed algorithm can be directly extended to construct this expanded basis, $\Phi_p(\cdot)$,  {by including  both the odd and even moments when generating the Hankel matrix \eqref{eq:Hankel_mtx}}.
For the exponential distribution where $\lambda=1$, the orthonormal polynomials for the basis $\phi_p$ (having only even degrees) are
\begin{align}\label{eq:exponential_op}
 \phi_1(X) &= \frac{1}{\sqrt{2}} X,  \nonumber \\ 
 \phi_2(X) &= \frac{1}{\sqrt{432}}(|X|^2X-12X), \nonumber  \\ 
 \phi_3(X) &= \frac{1}{\sqrt{654}}\left(\frac{|X|^4X}{40}-\frac{11}{6}|X|^2X + 13X\right),
\end{align}
and $\Phi_p$ (having both even and odd degrees) are
\begin{align}\label{eq:extend_exponential_op}
 \Phi_0(X) &= 1, \nonumber \\ 
 \Phi_1(X) &= X-1, \nonumber  \\ 
 \Phi_2(X) &= \frac{1}{2}\left(|X|X-4X+2\right), \nonumber \\
 \Phi_3(X) &= \frac{1}{6}\left(|X|^2X-9|X|X+18X-6\right).
\end{align}
This example illustrates that, even with the same distribution, more than one valid orthonormal polynomial can exist, depending on the form of the basis function. Note that the latter polynomial set  {in \eqref{eq:extend_exponential_op}} is well known as the Laguerre polynomial \cite{koekoek1991differential_Laguerre}.

\subsection{Orthonormal Polynomials for QAM Signals}

 {Other particularly relevant signal distributions to consider are those of} 4QAM, 16QAM, 64QAM, and 256QAM, corresponding to digital modulation techniques used widely in modern communication systems.
Let the  {constellation (set) of} possible modulation symbols be defined as $\mathcal{S} = \{s_1,s_2,\dots ,s_i\}$, where $i \in \{4, 16,64, 256\}$ is the modulation order. 
When symbols  {are drawn} uniformly from a given  {constellation}, the $m$th moment  {has the closed-form}
\begin{align}
    \mu_m = \frac{\sum_{k=1}^i |s_k|^m}{i}.
\end{align}

 {For the particular case of 4QAM, the Hankel matrix in \eqref{eq:Hankel_mtx} is a rank-$1$ matrix.} 
 {This is perhaps most obvious when the four symbols are on the unit circle, in which case they all have a second-order moment (or power) of $1$ and thus the higher-order moments are also $1$.}
 {Put simply,}  {orthonormal} polynomials of the third-order or higher  {do} not exist for 4QAM.
That is, the higher-order polynomials in \eqref{eq:GIH} are
\begin{align}
    \phi_p(x) = \sum_{k=0}^{p-1}c_{p,k}x.
\end{align}
Thus,  {when transmitting 4QAM symbols}, the nonlinearity is  {captured solely} by the first-order polynomial $x$. 
 {Orthonormal polynomials for higher-order QAM constellations are summarized in Table~\ref{table:OP_mcs}.}
\begin{table*}[h]
\caption{Orthonormal polynomials of  {popular}  {digital modulation schemes}} \label{table:OP_mcs}
\centering
\begin{tabular}{|c|l|l|}
\hline
\begin{tabular}[c]{@{}c@{}}\bf{Modulation}\\\bf{Scheme}\end{tabular} & \bf{Moments} & \bf{Orthonormal Polynomials (rounded coefficients)}\\ \hline
4QAM                                                        &   $\mu_{2k}$ = 1   &    $\phi_1(x) = x  $      \\ \hline
16QAM                                                       &   \begin{tabular}[c]{@{}l@{}}  $\mu_{2k}= \frac{1}{16\cdot 10^k}\left\{ 4 \cdot 18^k+4 \cdot 2^k+ 8\cdot 10^k  \right\}, $ \\  $[\mu_2, \mu_4, \mu_6, \mu_8, \dots ]= [1,1.32,1.96, 3.1248 ,\cdots ] $ \end{tabular}  &          \begin{tabular}[c]{@{}l@{}}     $\phi_1(x) = x $\\ 
 $\phi_2(x) = \frac{1}{\sqrt{0.2176}}(|x|^2x-1.32x) $ \\ 
 $\phi_3(x) = \frac{1}{\sqrt{0.0542}}\left(|x|^4x-2.47|x|^2x + 1.30x\right) $ \end{tabular} \\ \hline 
64QAM                                                       &  \begin{tabular}[c]{@{}l@{}} $\mu_{2k} = \frac{1}{64\cdot 42^k} \left\{ 4\cdot \left( 2^k +18^k+98^k \right) + 12\cdot 50^k  \right.$ \\ 
$\quad \quad +\left. 8\cdot \left( 10^k+26^2 +34^k +58^k+74^k\right )\right\}, $ \\ 
$[\mu_2, \mu_4, \mu_6, \mu_8, \dots ]= [1,1.381,2.2258, 3.9630 ,\cdots ] $\end{tabular}      &         \begin{tabular}[c]{@{}l@{}}     $\phi_1(x) = x $\\ 
 $\phi_2(x) = \frac{1}{\sqrt{0.3188}}(|x|^2x-1.381x) $ \\ 
 $\phi_3(x) = \frac{1}{\sqrt{0.1421}}\left(|x|^4x-2.7898|x|^2x + 1.6268x\right) $ \end{tabular} \\ \hline 
256QAM                                                      &   \begin{tabular}[c]{@{}l@{}} $\mu_{2k} = \frac{1}{256}\sum_{i=1}^{256} |s_i|^{2k} $ \\$[\mu_2, \mu_4, \mu_6, \mu_8, \dots ]= [1,1.3953,2.2922, 4.1910 ,\cdots ] $  \end{tabular}   &                \begin{tabular}[c]{@{}l@{}}     $\phi_1(x) = x $\\ 
 $\phi_2(x) = \frac{1}{\sqrt{0.3453}}(|x|^2x-1.3953x) $ \\ 
 $\phi_3(x) = \frac{1}{\sqrt{0.1772}}\left(|x|^4x-2.8747|x|^2x + 1.7189x\right) $ \end{tabular} \\ \hline 
\end{tabular}
\end{table*}


\section{SIC using Adaptive Orthonormal Polynomials}

In this section, we  {introduce} a LMS-based SIC algorithm using orthonormal polynomials. 
We first propose an algorithm that can be used universally for arbitrary signals, that is,  {when moments are not known a priori and must be computed.}
 {We} then propose  {a more implementation}-friendly algorithm  {which leverages} a look-up table when the moments of signals are known  {a priori},  {such as when employing predefined MCSs, like in 5G and Wi-Fi.}

\subsection{LMS Algorithm using Adaptive Orthonormal Polynomials}

As a first step,  {our} proposed SIC algorithm estimates the moments of the transmit symbols $x[n]$. The $p$th moment of the transmit symbol $x[n]$ is estimated by its sample average:
\begin{align}
    \mu_p = \frac{1}{N}\sum_{n=0}^{N-1} |x[n]|^p.
\end{align}
  {This estimation converges to the true moment as the sample size $N$ increases by the law of large numbers, i.e.,}
\begin{align}
     \lim_{N\rightarrow \infty}\frac{1}{N}\sum_{n=0}^{N-1} |x[n]|^p = \mathbb{E}\left[|x[n]|^p\right].
\end{align}
 {Even for modest $N$}, this sample average estimator can  {often closely approximate the true moment} since  estimation errors reduce linearly with sample size {\cite{Ueda96generalizationError}}. 
Then, the basis functions have orthogonality by  {constructing} the coefficients  {according to} Algorithm~\ref{alg:orthonormal_polynomial} as ${\bm \phi}^{\sf AOP}({\bf x}[n]; \Theta_1)= \left[\phi^{\sf AOP}_{1}(x[n],{\hat{\bf c}}_1),\ldots, \phi^{\sf AOP}_{\frac{P+1}{2}}(x[n],\hat{{\bf c}}_{\frac{P+1}{2}})\right]^{\top}$.
Theorem~\ref{thm:approx-PA} is provided  {below} to show that  {SI} can indeed be canceled using the proposed orthonormal polynomial basis functions.
\begin{thm}\label{thm:approx-PA}
     {A nonlinear PA response function $f(x)$ approximated by a sum of polynomials can be expressed as}
    \begin{align}
        f(x[n]) = \sum_{k=1}^{\frac{P+1}{2}} \pmb{\phi}^{\sf OP}_k(\mathbf{x}[n];\mathbf{c}_k)^{\sf H}\mathbf{h}_k[n],
    \end{align}
    where $\mathbf{h}_k[n]$ is a weight vector and $\pmb{\phi}^{\sf OP}_k(\mathbf{x}[n];\mathbf{c}_k)$ is an orthonormal basis function whose highest order is $2k-1$ and has coefficients $\mathbf{c}_k$.
\end{thm}

\begin{proof}
    See Appendix 2.
\end{proof}

To derive the weights $\{\mathbf{h}_k[n]\}$ which satisfy Theorem~\ref{thm:approx-PA}, we introduce an adaptive method with our proposed orthonormal basis function.
The LMS algorithm is a stochastic approximation of the iterative steepest descent based implementation of the Wiener filter and is applicable when the  {SI} and the input signal are jointly wide-sense stationary (JWSS) \cite{Haykin_Book2008,sayed2011adaptive,widrow1977stationary}.
This stochastic approximation involves a simple update equation that can be implemented in practical systems with low computational complexity. Using the classical LMS algorithm \cite{Haykin_Book2008}, the linear filter parameter coefficients are updated as
\begin{align}
{\bf \hat h}[n+1]= {\bf \hat h}[n] +\mu  {\bm \phi}^{\sf AOP}({\bf x}_L[n]; \Theta_1) e^*[n], \label{eq:LMS}
\end{align}
where $\mu$ is the step size.


{\bf Remark 1 on Computational Complexity:} The computational cost associated with implementing orthonormal polynomials in the AOP-LMS algorithm is $\mathcal{O}\left(\frac{P^3}{8}\right)$, where the majority of this complexity involves constructing the orthonormal basis functions.

\begin{figure}[h]
    \centering
    \includegraphics[scale=0.35]{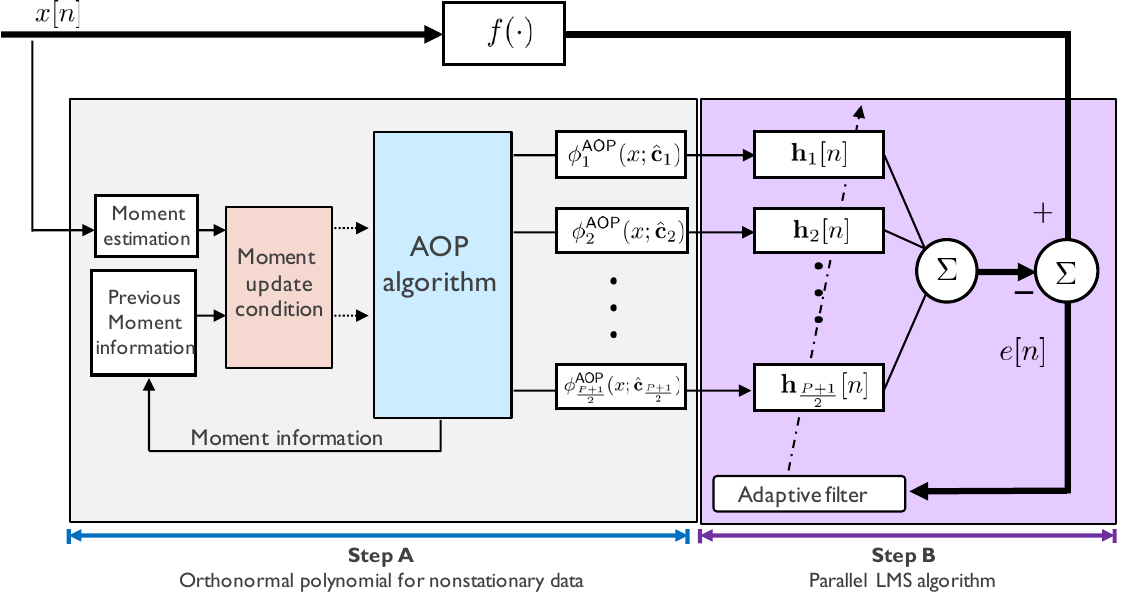}
    \caption{Block diagram of the  {proposed} AOP-LMS algorithm  {for} non-stationary input signals.}
    \label{fig:Non_stationary_AOP_diagram}
\end{figure}

 \begin{figure}[h]
    \centering
    \includegraphics[scale=0.35]{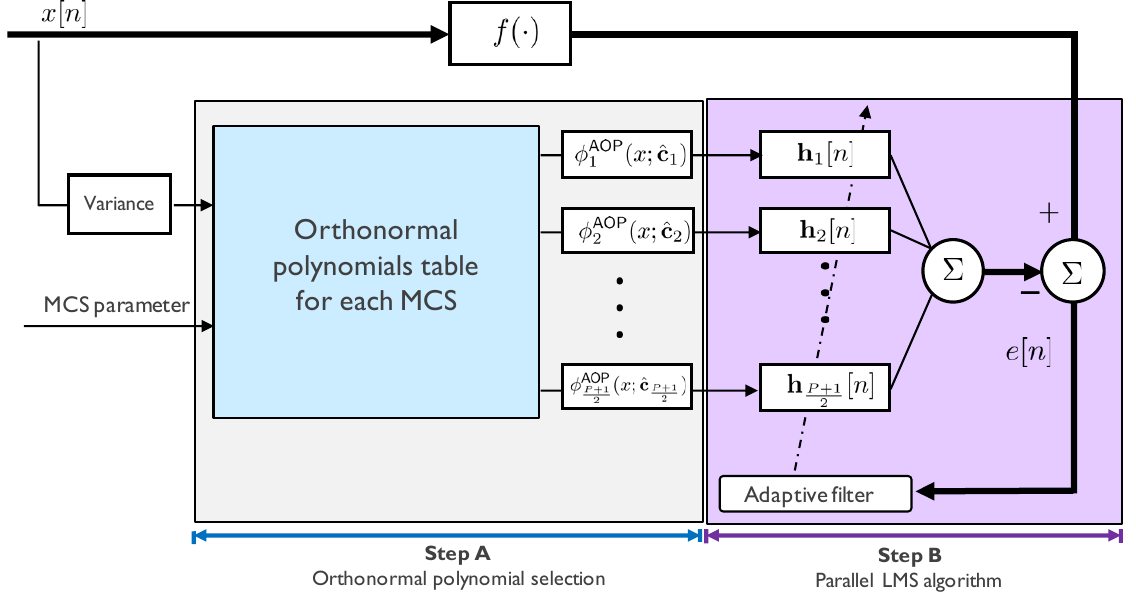}
    \caption{Block diagram of the  {proposed} LMS algorithm using a LUT.  {Pre-computed orthonormal polynomials can be retrieved from the LUT as the MCS changes during link adaptation in order to reduce computational complexity at run-time.}}
    \label{fig:LUT_diagram}
\end{figure}

\subsection{Extending Our Algorithm to Non-Stationary Input Signals}

\textbf{Filter Adaptation:} 
Based on adaptive filter theory \cite{Haykin_Book2008}, the performance of the LMS algorithm may vary depending on the conditioning of the covariance matrix of the input basis functions. 
That is, in an environment in which the distribution changes with time, using  {a fixed set of} basis functions  {may cause degradation in} SIC performance.  
 {This motivates the} need for a method which generates  {orthonormal} basis functions that  {adapt} to changes in the distribution of the input signal. 
Fig.~\ref{fig:Non_stationary_AOP_diagram} depicts a block diagram of one such method. 
 {In this proposed method, the estimated moments of the transmit signal are updated by comparing the previous moment estimates to the current estimates, updating them based on some specified condition.}
Technically speaking, a change in the moments necessitates an update of the basis functions, but one could capture slight changes in the signal distribution by instead updating $\{\mathbf{h}_k[n]\}$. 
This avoids recomputing the polynomial basis functions and instead leverages the computational simplicity of adaptive algorithms, such as LMS.



As one example of such an approach, Algorithm~\ref{alg:Non-LMS} describes a method which uses nonlinear LMS to  {regularly} update the  {estimated moments} only in a predetermined interval.
{In practice, this interval could correspond to the duration of time slots defined by wireless standards, such as 5G and Wi-Fi.} 
We define $N_{\mathrm{max}}$ as the maximum number of samples used to estimate the transmit signal's statistics, $N_\mathrm{int}$ as the sample interval at which statistical information is fixed, and $n_\mathrm{s}$ as the start sample for statistical information collection. 

\textbf{LUT Adaptation:} 
 {As mentioned before, in practical wireless systems, transmit signals are non-stationary due to link adaptation, whereby a transmitter adapts its MCS according to the channel strength/quality.}
 {Since the set of possible MCSs are known a priori, the transmit signal's statistics and thus the orthonormal polynomials can be pre-computed for each MCS.
Storing these pre-computed polynomials in a look-up table (LUT) and then referencing them at run-time can reduce computational complexity.}
 {Our proposed method employing such a technique is depicted in Fig.~\ref{fig:LUT_diagram}.}



 {In Table~\ref{table:Comparison_table}, we compare existing digital SIC algorithms against our proposed technique in terms of model complexity and performance.
While the HP-W-LMS algorithm and our proposed AOP technique are on par with one another in terms of complexity, we will see in the next section that ours offers superior robustness/adaptation to changes in the transmit signal distribution and in the SI propagation channel.}


\begin{algorithm}[h]
 \caption{Digital SIC for non-stationary  {transmit signals} using nonlinear LMS}\label{alg:Non-LMS}
 \begin{algorithmic}[1]
  \State \textbf{Input:} $x[n]$, $P$, $N_{\mathrm{max}}$, and $N_{\mathrm{int}}$
  \State \textbf{Output: } $e[n]$
  \State \textbf{Initialize: } $n_\mathrm{s} \leftarrow 1$
 \For{$n = 1$ to $N$}
    \If{$n \leq n_\mathrm{s}+N_{\mathrm{max}}-1$}    
       \State \bf{Nonlinear basis }$\pmb{\phi}_p^{\sf AOP}(\mathbf{x}[n];\hat{{\mathbf{c}}}_p)$ \bf{generation via Algorithm}~\ref{alg:orthonormal_polynomial}
    \EndIf
    \If{$n = (n_\mathrm{s}+N_{\mathrm{int}}-1)$}
        \State $n_\mathrm{s} \leftarrow n$
    \EndIf
    \State $e[n] \leftarrow y[n]-\sum_{p=1}^\frac{P+1}{2} \mathbf{h}_p^{\sf H} \pmb{\phi}_p^{\sf AOP}(\mathbf{x}[n];\hat{{\mathbf{c}}}_p)$
    \State $\mathbf{h}_p[n+1] \leftarrow \mathbf{h}_p[n] + \mu_p e^*[n]\pmb{\phi}_p^{\sf AOP}(\mathbf{x}[n];\hat{{\mathbf{c}}}_p)$
 \EndFor
   \end{algorithmic}
\end{algorithm}

\begin{table}[h] 
\caption{Comparison of LMS-based digital SIC algorithms.}
\label{table:Comparison_table}
\resizebox{\columnwidth}{!}{ 
\begin{tabular}{|c|c|c|c|c|}
\hline
\begin{tabular}[c]{@{}c@{}} \bf{SI channel}\\ \bf{model}\end{tabular}    & \begin{tabular}[c]{@{}c@{}} \bf{Model}\\ \bf{complexity}\end{tabular} & \begin{tabular}[c]{@{}c@{}}\bf{Adaptation}\\ \bf{complexity}\end{tabular}  & \begin{tabular}[c]{@{}c@{}}\bf{Performance}\\ \bf{(+speed)}\end{tabular}         & \begin{tabular}[c]{@{}c@{}}\bf{Nonstationary}\\ \bf{distribution}\end{tabular} \\ \hline
\begin{tabular}[c]{@{}c@{}}Linear\\ +Wiener filter\end{tabular}      & $M$     & $\mathcal{O}(1) $        & \begin{tabular}[c]{@{}c@{}}(Only linear)\\ Limited/Fast\end{tabular}   & No   \\ \cline{1-5}  
\begin{tabular}[c]{@{}c@{}}Hammerstein\\ polynomial\\ +Wiener filter\end{tabular} & $ \frac{P+1}{2}M$      &   $\mathcal{O}(1)$   & Limited/Slow     & No      \\ \cline{1-5}  
\begin{tabular}[c]{@{}c@{}}IH polynomial\\ +Wiener filter\end{tabular}            &  $\frac{(P+1)(P+3)}{8}M$ &    $\mathcal{O}(1)$       & \begin{tabular}[c]{@{}c@{}}(Conditionally)\\ Optimal/Fast\end{tabular} & No    \\\cline{1-5}  
\begin{tabular}[c]{@{}c@{}}HP+Whitening\\ +Wiener filter\end{tabular}  & $\left( \frac{P+1}{2}\right)^2M$   & $\mathcal{O}(P^3)$   & Optimal/Fast   & Yes    \\ \cline{1-5}  
\begin{tabular}[c]{@{}c@{}}AOP\\ +Wiener filter\end{tabular}                      &   $\frac{(P+1)(P+3)}{8}M$ & $\mathcal{O}(P^3)$  & Optimal/Fast    & \begin{tabular}[c]{@{}c@{}}Yes\\ (More robust)\end{tabular}          \\ \hline
\end{tabular}
}
\end{table}

\section{Simulation Results} \label{sec:simulation}

In this section, we provide simulation results to verify the  {effectiveness} of the proposed SIC algorithm. 
In our simulations, we will consider both stationary and non-stationary transmit signals.
We consider nonlinear distortion based on \eqref{eq:Saleh_model}
where the small signal gain  {is} $\gamma = 3$, the transition sharpness  {is} $\beta = 0.09$, and the memory effect coefficients $\{h[n]\}$  {are} drawn from the complex Gaussian distribution with length $M=9$.
In addition, we assume that the noise floor is  {$-100$~dBm} \cite{Sabharwal:2014band}. 
 {To explore performance across input distributions, transmitted signals are generated according to a variety of distributions including Gaussian, QAM, and mixtures of such.}
 {After running our proposed digital SIC technique,} we measure its performance in the mean squared error (MSE) sense as
\begin{align}
    {\sf{MSE}} = \mathbb{E}\left[\left\vert y[n]-\hat{f}(\mathbf{x}[n])\right\vert^2\right].
\end{align}

In the simulation results that follow, we compare the proposed AOP-LMS algorithm against existing SIC techniques, including HP-LMS \cite{Book:ding2004digital}, HP-W-LMS (or pre-orthogonalization LMS) \cite{Korpi:2015}, and IH-LMS \cite{Jungyeon:2018}. 
In addition, we also compare our proposed algorithm against model-free SIC techniques, such as kernel LMS \cite{liu2008kernel} and the neural-network-based cancellation \cite{Alexios2018neural} for the stationary case.

\begin{figure}
    \centering
    \includegraphics[width=0.45\textwidth]{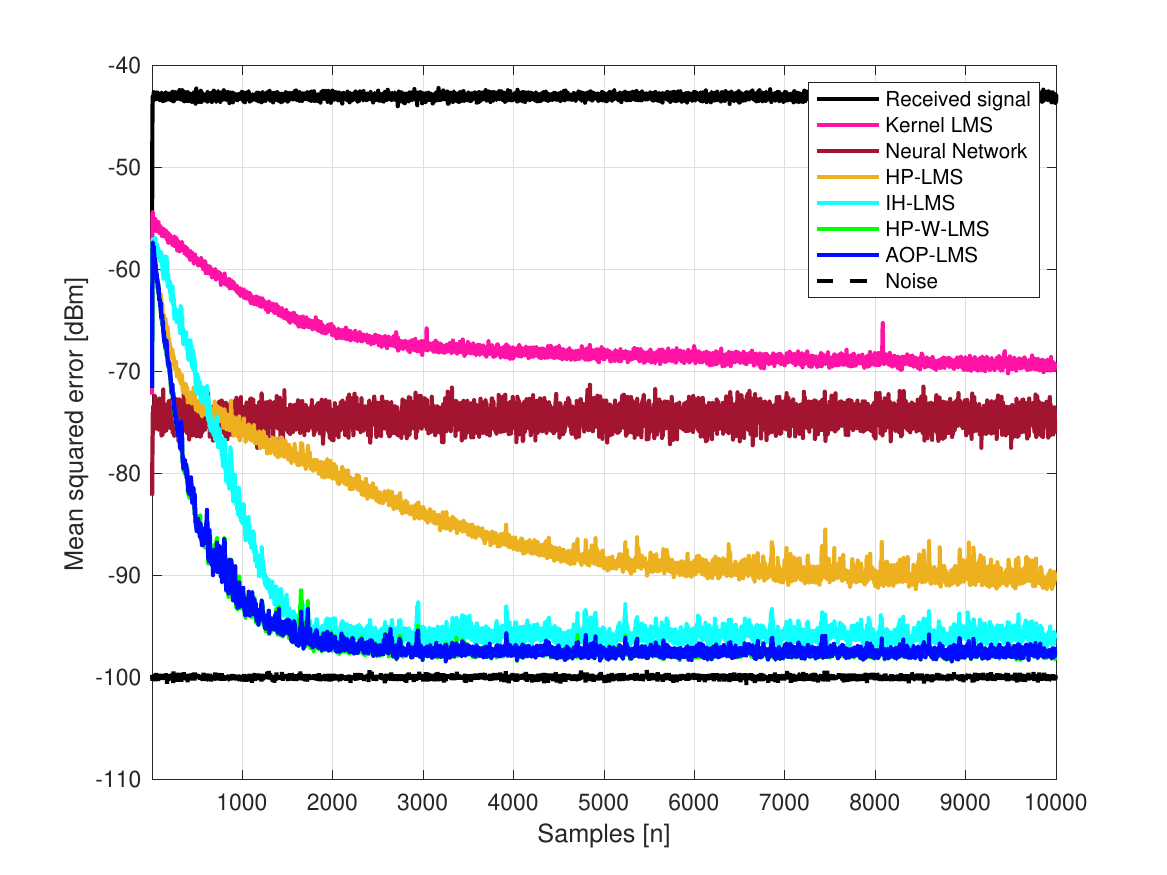}
    \caption{MSE performance of different SIC algorithms when $P=7$. The transmit signal follows the  {Gaussian-QAM mixture}.}
    \label{fig:mixed_MSE}
\end{figure}
\begin{figure}
    \centering
    \includegraphics[width=0.45\textwidth]{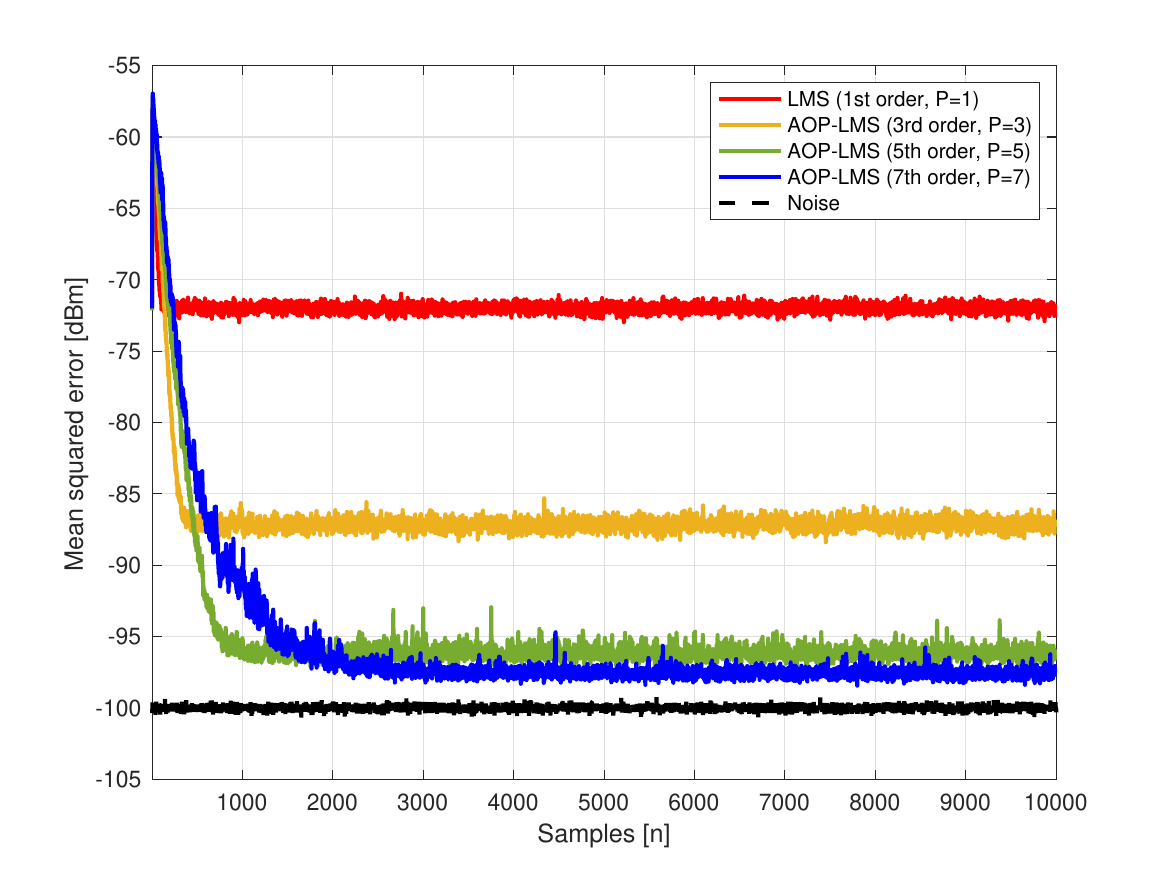}
    \caption{MSE performance of AOP-LMS  {for varying} polynomial orders $P\in \{1,3,5,7\}$. The transmit signal follows the Gaussian-QAM mixture.}
    \label{fig:mixed_MSE_order}
\end{figure}

\begin{figure}
    \centering
    \includegraphics[width=0.45\textwidth]{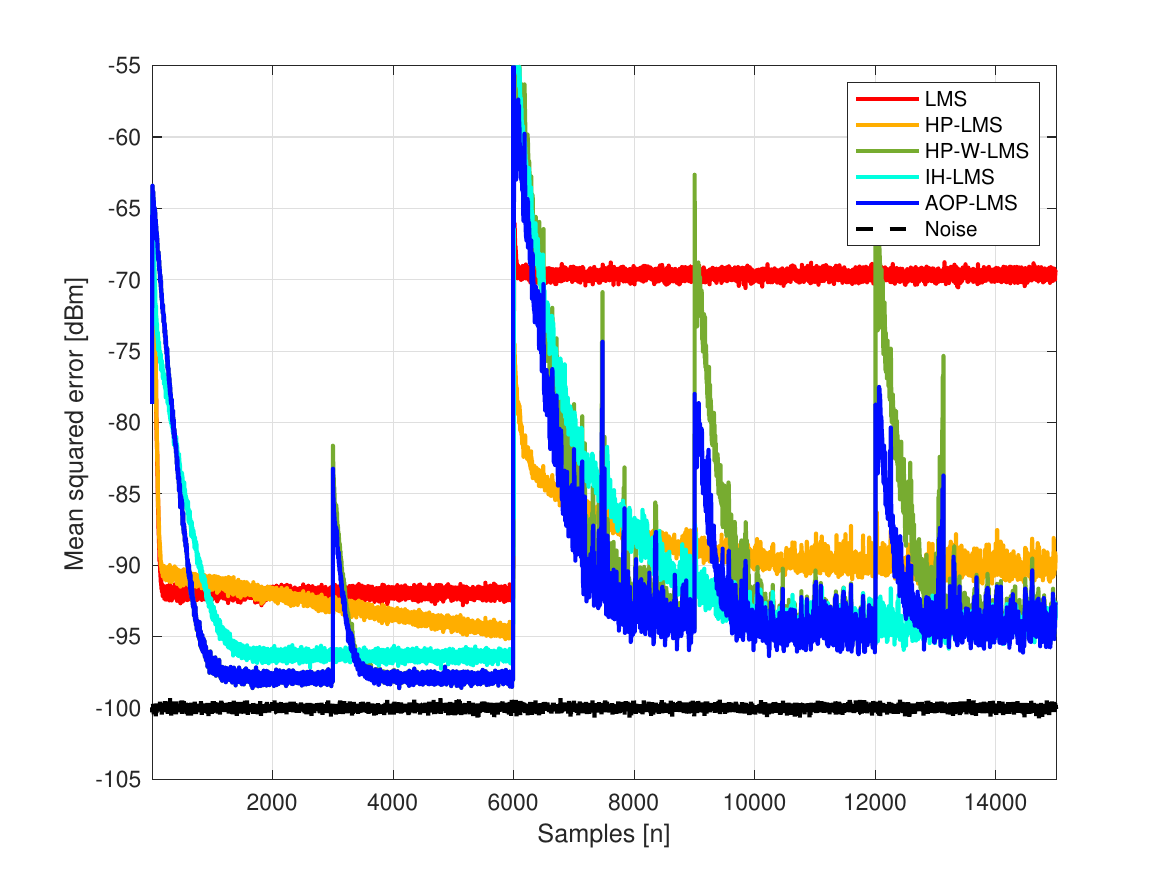}
    \caption{MSE performance of various SIC algorithms for non-stationary transmit signals.}
    \label{fig:nonstationary}
\end{figure}
\begin{figure}[h]
    \centering 
    \subfigure[QAM symbols with OFDM signaling.]{
    \includegraphics[width=0.45\textwidth]{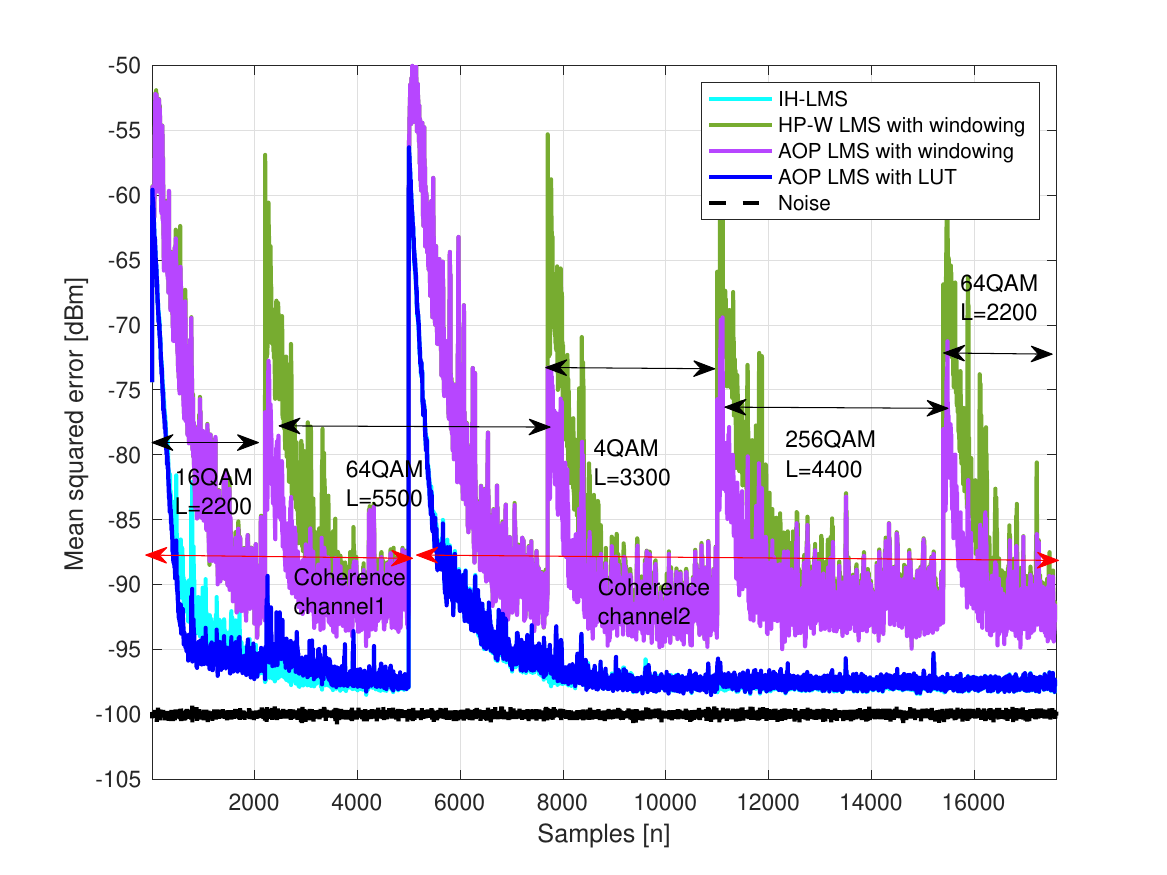}
    \label{fig:nonstationary_OFDM}
    }
    \subfigure[QAM symbols with SC-FDE signaling.]{
    \includegraphics[width=0.45\textwidth]{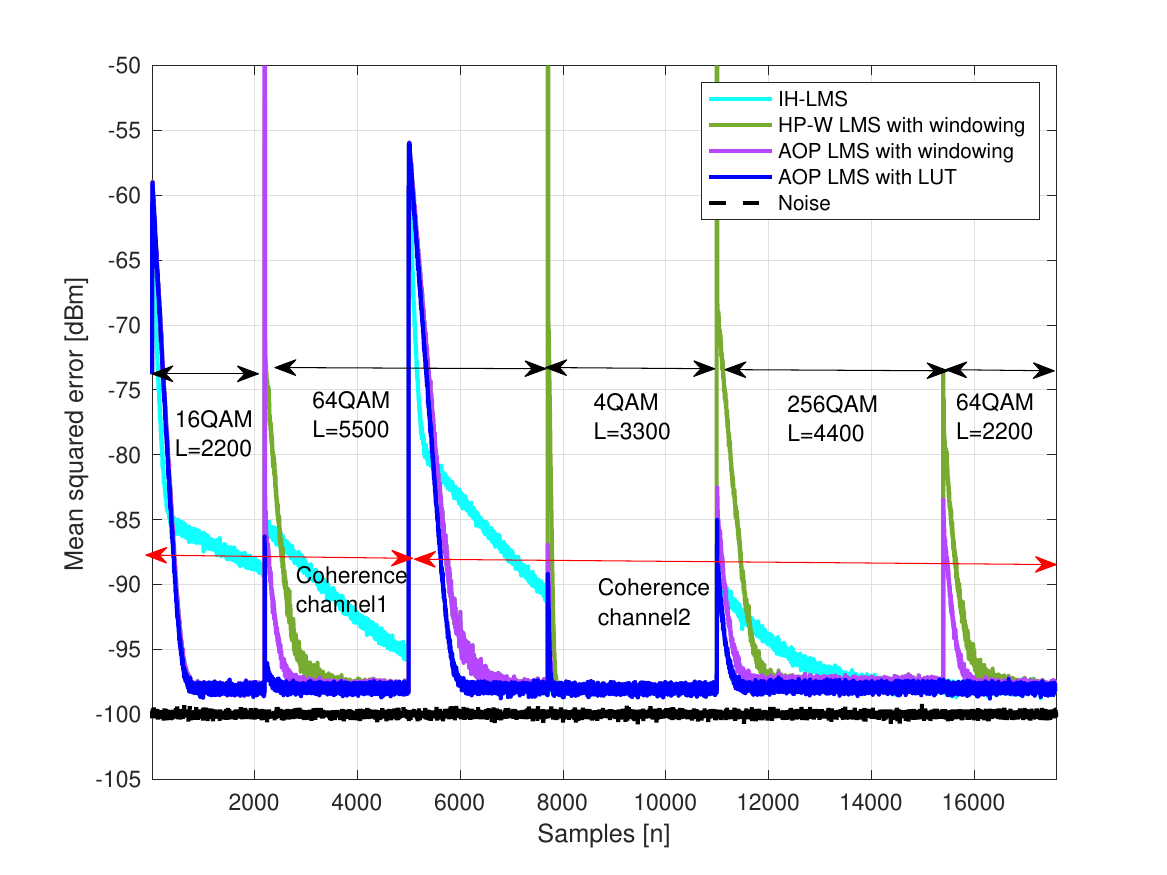}
    \label{fig:nonstationary_SCFDE}
    }
    \caption{
     {MSE achieved by various SIC techniques with QAM symbols under (a) OFDM signaling and (b) SC-FDE signaling. With OFDM, the resulting signal is approximately Gaussian by CLT, whereas under SC-FDE it is not and hence our proposed AOP-LMS scheme is superior.}
    }
    \label{fig:nonstationary_application}
\end{figure}

{\bf Stationary Transmit Signals:} Let $x_c[n]\sim \mathcal{CN}(0,1)$ be a complex Gaussian distribution with zero mean and variance $1$, and let $x_q[n]$ be a $4$ {QAM} signal as  {$x_q[n] \in \frac{1}{\sqrt{2}}\{1+\mathrm{j},1-\mathrm{j},-1+\mathrm{j},-1-\mathrm{j}\}$}. 
We assume that the signal is generated as the mixture $x[n] = x_c[n]+ x_q[n]$ and that its distribution does not change over time. 
Unlike the specific distributions in Section 4 that have well-known  {orthonormal} polynomials, the  {orthonormal} polynomials for this signal are unknown.
Nonetheless, as shown in Fig.~\ref{fig:mixed_MSE}, our AOP-LMS algorithm achieves the lowest MSE, on par with  {HP-W-LMS} and shows faster convergence than the other techniques.
 
Fig.~\ref{fig:mixed_MSE_order} shows the MSE performance of the proposed AOP-LMS algorithm for various polynomial orders $P$.
As expected, with sufficiently high-order polynomials, the proposed algorithm converges to a lower MSE. 
However, as the order increases, the number of filters increases, and thus, the rate of convergence slows slightly.


{\bf Non-Stationary Transmit Signals:} To evaluate the case of non-stationary transmit signals, we simulate the transmission of data following a uniform distribution with a range of $[-1,1]$ for the first 6000 samples, followed by the Gaussian-QAM mixture $x_c[n]+x_q[n]$ for the subsequent 9000 samples. 
The number of samples to learn the moments was set to  {$N_\mathrm{max} = 55$ (5\% of 1100 samples)}, with a basis initialization interval of 3000. 
The remaining 2945 samples were filtered using the orthonormal polynomial obtained from the empirically estimated moment information as the adaptive filter basis.

Fig.~\ref{fig:nonstationary} compares the MSE performance of different techniques for the aforementioned non-stationary transmit signal. 
 {The superiority of AOP-LMS and HP-W-LMS demonstrates that tailoring the polynomial basis functions to the transmit signal offers faster convergence and lower MSE compared to the other techniques, which merely fix their basis functions to an assumed underlying signal distribution.}
 {Moreover, notice that, when the signal distribution changes from uniform to the Gaussian-QAM mixture (beyond 6000 samples), AOP-LMS outperforms HP-W-LMS. 
This stems from the fact that, with a limited number of samples, HP-W-LMS can suffer from poor conditioning of its sample covariance matrix. 
In turn, this leads to ineffective whitening in pre-orthogonalization and subsequent spikes in MSE when triggering re-estimation of the signal statistics. 
AOP-LMS, on the other hand, is only affected by this poor conditioning in its higher-order polynomials, and as a result, its lower-order polynomials allow it to maintain cancellation by preserving some prior knowledge of the signal distribution and channel.} 





{\bf  {Adaptive Modulation and Coding}:}  {To combat} the  {time-variability of wireless channels due to fading},  {communication systems adapt their signal distributions over time according to the channel strength}. 
To simulate this, generate a transmit signal consisting of a sequence of 16QAM symbols, followed by 64QAM symbols, then 4QAM, then 256QAM, and finally 64QAM, whose corresponding durations are 2200, 5500, 3300, 4400, and 2200 samples. 
 {In addition, to demonstrate that our approach can adapt to changes in the channel, we draw an independent realization of the channel at the 5000th sample.} 
 {We consider transmission of these QAM symbols using two commonly used waveforms: OFDM and  {single-carrier frequency domain equalization (SC-FDE)}.
With OFDM, the QAM symbols are treated as frequency-domain symbols and are transformed to the time domain using the FFT. 
With SC-FDE, the QAM symbols are transmitted directly using conventional single-carrier modulation.}


 {Fig.~\ref{fig:nonstationary_application} shows the MSE attained with various SIC algorithms  {for} (a) OFDM and (b) SC-FDE.
In Fig.~\ref{fig:nonstationary_OFDM}, it can be seen that, under OFDM signaling, the proposed AOP-LMS technique using the LUT achieves impressive MSE performance and fast convergence, on par with the IH-LMS. 
Leveraging  the LUT allows the proposed method to adapt quickly to changes in the constellation, when compared to those that rely on run-time estimation of the signal statistics. 
It is well known that the distribution of OFDM signals are approximately Gaussian by CLT \cite{Jungyeon:2018}, and thus, IH-LMS attains  performance comparable to the proposed method.
In Fig.~\ref{fig:nonstationary_SCFDE}, on the other hand, we see that IH-LMS no longer offers adequate cancellation under SC-FDE signaling. 
This is because the QAM symbols are transmitted directly, without an FFT applied, and thus are no longer approximately Gaussian. 
As a result, the proposed AOP-LMS technique offer the lowest MSE and the fastest convergence, especially when leveraging the LUT with pre-computed polynomials. 
For the particular case of 4QAM symbols, we see comparable performance across all four schemes; this is courtesy of the fact that the basis is purely linear, as highlighted previously.}





\section{Experiments with a Wireless Testbed}

In this section, we present experimental results  {collected with} a single-input single-output wireless testbed. 
To accomplish this, we  {applied} our algorithms on real-world data collected with software-defined radio (SDR) platforms  {developed} by National Instruments.
A summary of implementation details is shown in Table \ref{USRP:spec}, and a photo of the hardware platform is shown in Fig.~\ref{fig:PXI_chassis}.


\begin{figure}
    \centering
    \includegraphics[width=0.35\textwidth]{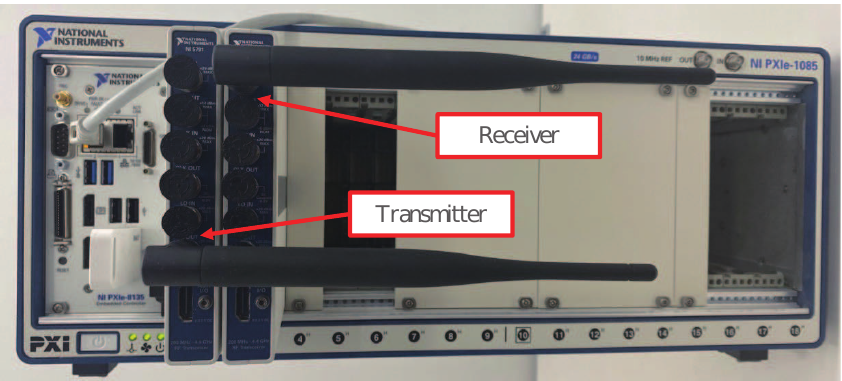}
    \caption{ {The NI PXIe-1085 chassis and NI-5791 modules used to experimentally validate the proposed SIC algorithms. Separate antennas were used for transmission and reception.}}
    \label{fig:PXI_chassis}
\end{figure}

\begin{table}[t]
\centering
\caption{Specifications of the  {experimental} implementation} \label{USRP:spec}
\begin{tabular}{|cc|}
\hline
\multicolumn{2}{|c|}{\textbf{Testbed Parameters}}                                   \\ \hline
\multicolumn{1}{|c|}{Signal bandwidth}            & 10 MHz                      \\ \hline
\multicolumn{1}{|c|}{Center frequency}            & 2.45 GHz                    \\ \hline
\multicolumn{1}{|c|}{Sampling rate}               & 120 MHz                     \\ \hline
\multicolumn{1}{|c|}{Upsampling rate}               & 12                      \\ \hline
\multicolumn{2}{|c|}{\textbf{SIC Algorithm Parameters}}                                   \\ \hline
\multicolumn{1}{|c|}{Highest  {nonlinearity} order}           & 5                          \\ \hline
\multicolumn{1}{|c|}{The number of filter taps}   & 11                         \\ \hline
\end{tabular}
\end{table}

\begin{figure}[h]
    \centering 
    \subfigure[Constellation (complex Gaussian symbols).]{
    \includegraphics[width=0.41\textwidth]{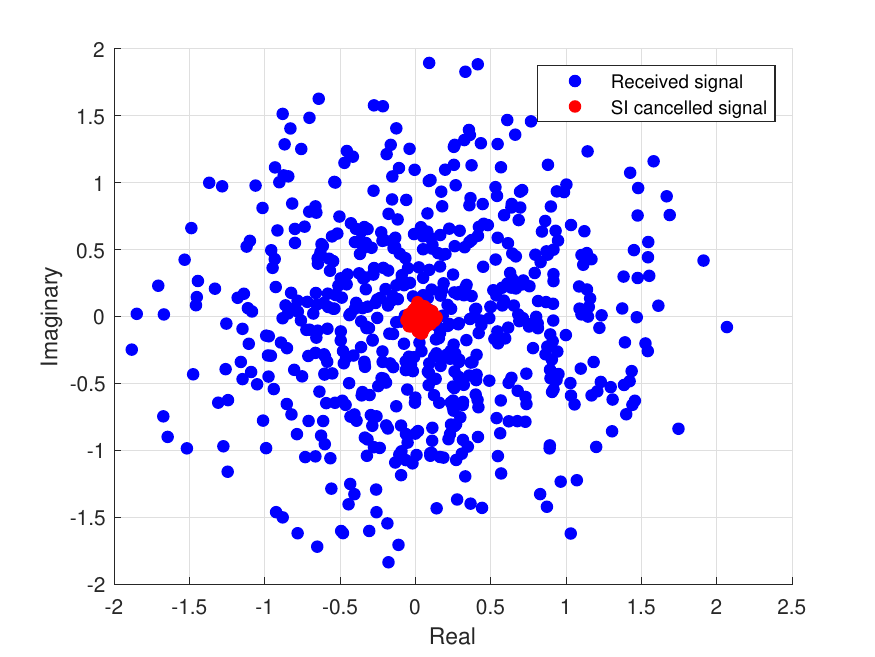}
    \label{fig:USRP_Gaussian_scatter}
    }
    \subfigure[Power spectral density (complex Gaussian symbols).]{
    \includegraphics[width=0.41\textwidth]{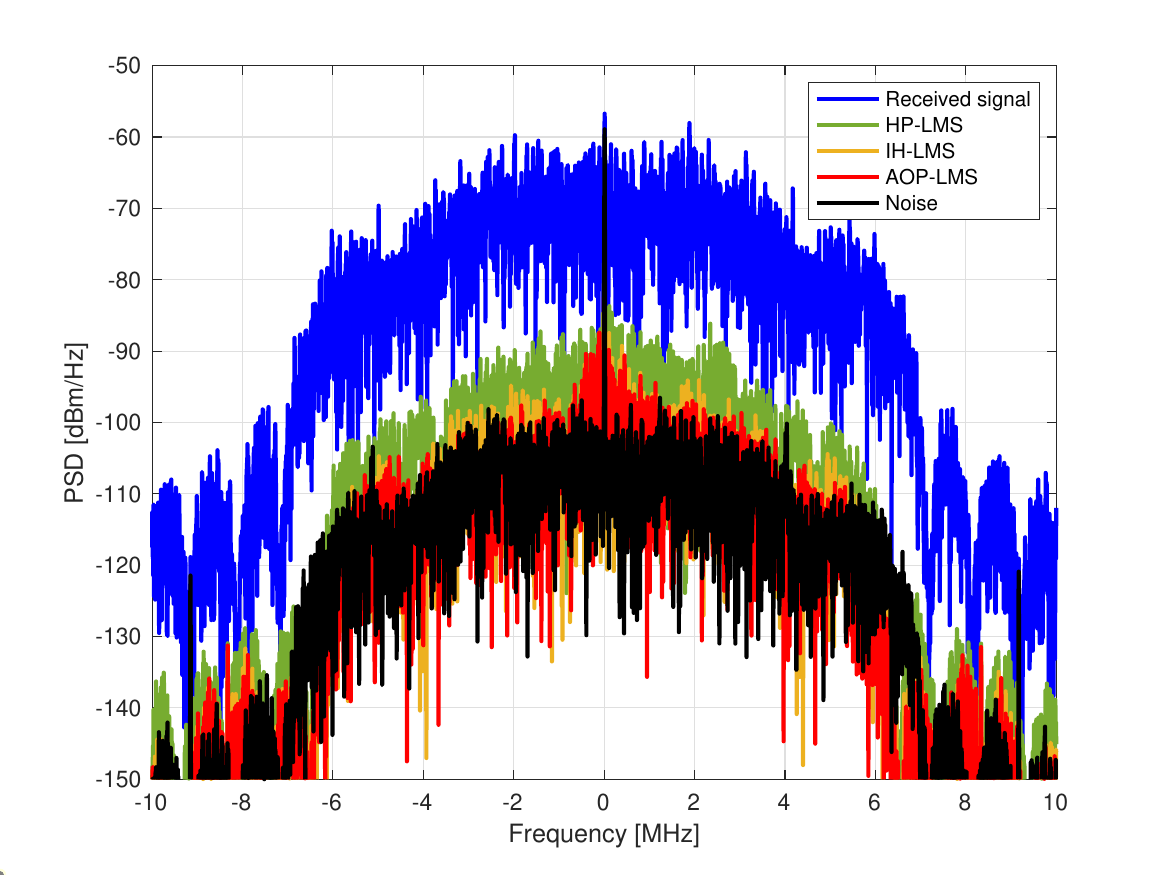}
    \label{fig:USRP_Gaussian_PSD}
    }
    \caption{
    (a) A constellation of the received signal (blue) and the residual SI (red). 
    (b) The power spectral density of various SIC algorithms for a $10$ MHz waveform at $2.45$ GHz. 
    In both, complex Gaussian symbols were transmitted.
    }\label{fig:USRP_Gaussian_}
\end{figure}

\begin{figure}[h]
    \centering 
    \subfigure[Constellation (uniform symbols).]{
    \includegraphics[width=0.41\textwidth]{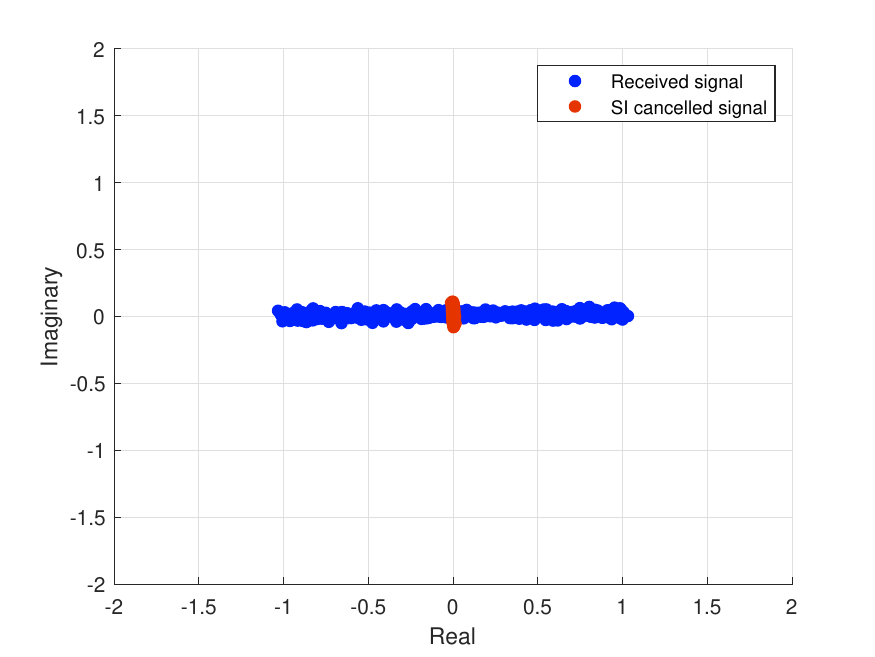}
    \label{fig:USRP_uniform_scatter}
    }
    \subfigure[Power spectral density (uniform symbols).]{
    \includegraphics[width=0.41\textwidth]{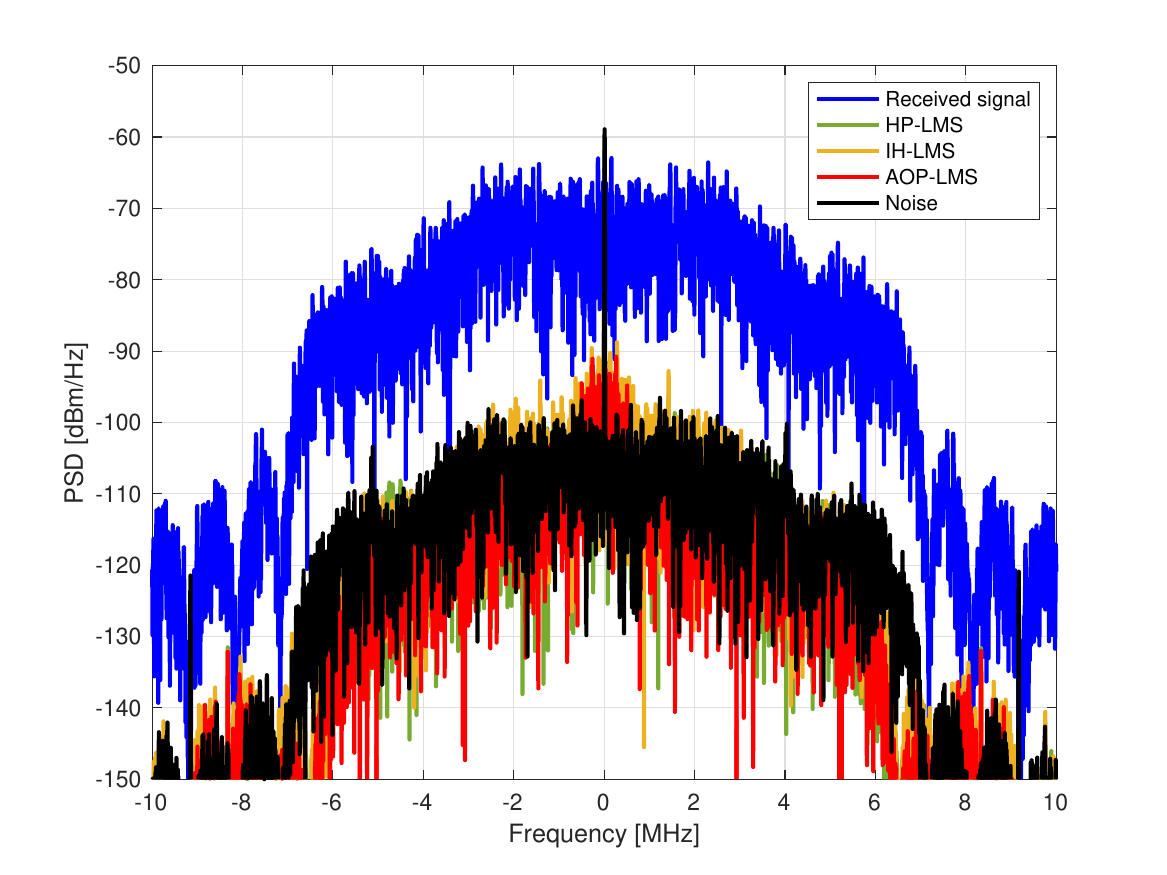}
    \label{fig:USRP_uniform_PSD}
    }
    \caption{
    (a) A constellation of the received signal (blue) and the residual SI (red). 
    (b) The power spectral density of various SIC algorithms for a $10$ MHz waveform at $2.45$ GHz. 
    In both, uniform symbols were transmitted.
    }\label{fig:USRP_Uniform_}
\end{figure}

\begin{figure}[h]
    \centering 
    \subfigure[Constellation (Gaussian-QAM mixture).]{
    \includegraphics[width=0.41\textwidth]{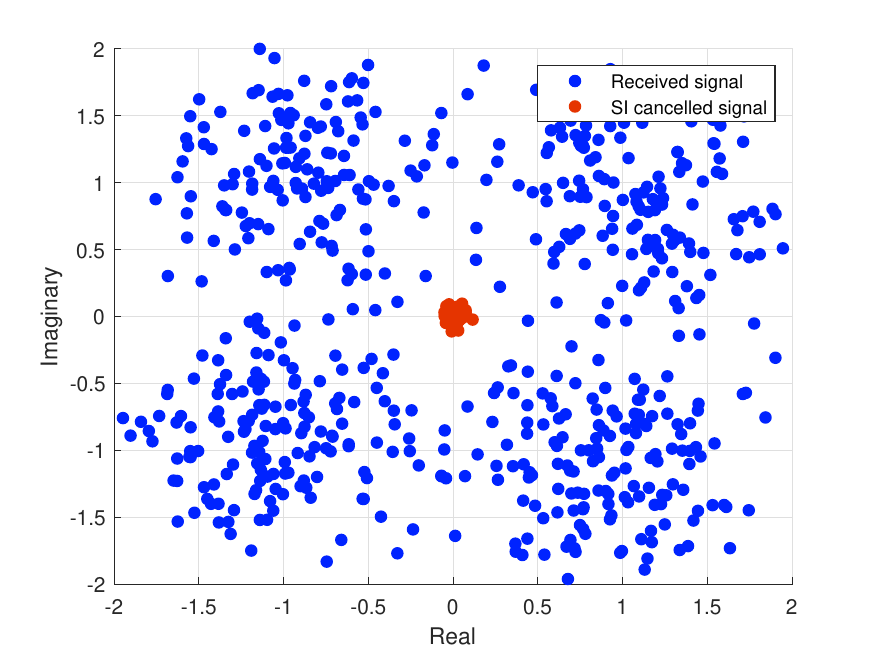}
    \label{fig:USRP_Gaussian_Qam_scatter}
    }
    \subfigure[Power spectral density (Gaussian-QAM mixture).]{
    \includegraphics[width=0.41\textwidth]{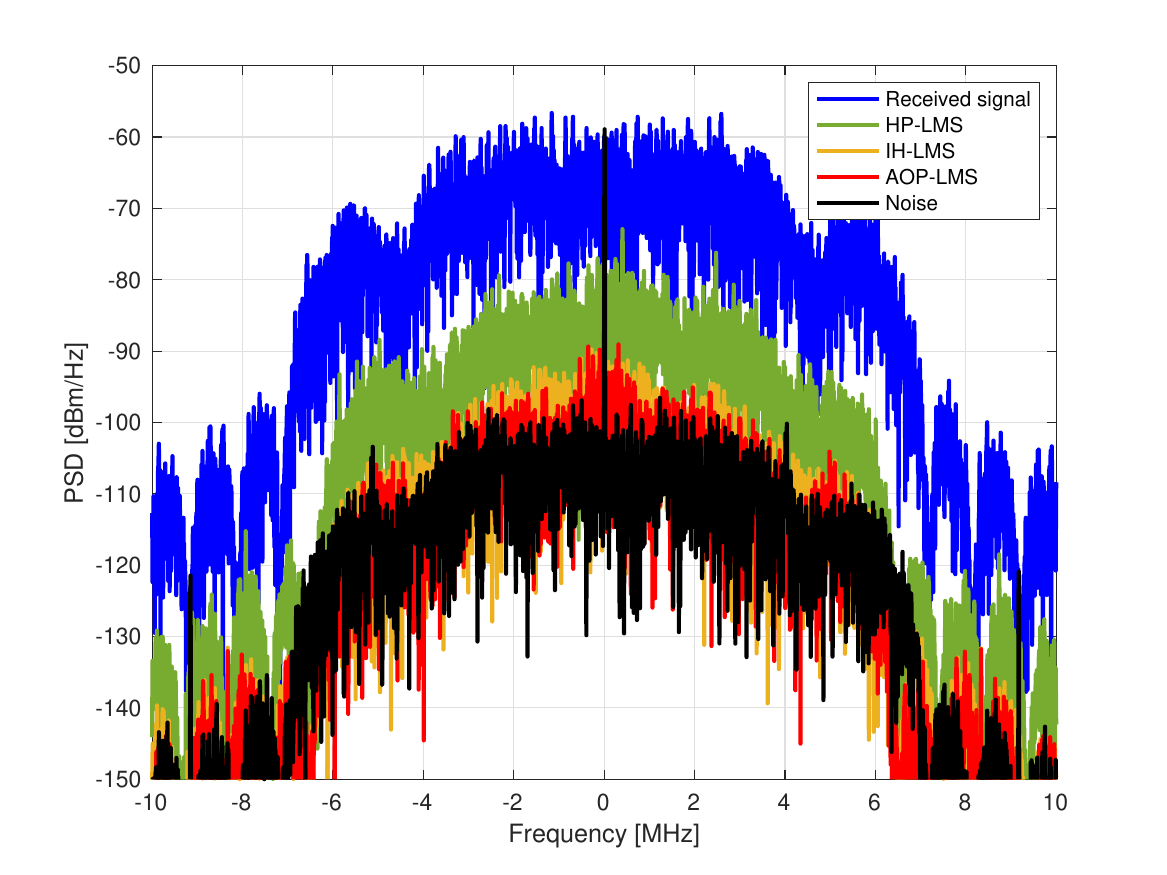}
    \label{fig:USRP_Gaussian_Qam_PSD}
    }
    \caption{
    (a) A snap shot of the received signal and the SI-cancelled signal of the mixed distribution,
    (b) Power spectral density of various SIC algorithms with a $10$ MHz waveform, measured at $2.45$ GHz of the mixed distribution.
    }\label{fig:USRP_Gauss_qam_}
\end{figure}

Figs.~\ref{fig:USRP_Gaussian_}--\ref{fig:USRP_Gauss_qam_} show the IQ  {scatter} plot and power spectral density (PSD) for various signal distributions: (i) complex Gaussian, (ii) uniform, and (iii)  {a Gaussian-QAM mixture}. 
Each IQ scatter plot shows the constellation of the transmitted signal (in blue) alongside the residual SI signal after cancellation with the proposed AOP-LMS technique (without LUT). 
Across all three, Figs.~\ref{fig:USRP_Gaussian_}--\ref{fig:USRP_Gauss_qam_}, we see that AOP-LMS cancels SI by about $30$~dB to near the noise floor.
The robustness of AOP-LMS is clear, as it is the most capable of canceling SI across all three considered signal distributions. 

For the complex Gaussian distribution of Fig.~\ref{fig:USRP_Gaussian_}, IH-LMS is naturally on par with AOP-LMS, since it is designed specifically for Gaussian signals. 
For the uniform distribution of Fig.~\ref{fig:USRP_Uniform_}, HP-LMS is on par with AOP-LMS and slightly better than IH-LMS  {as shown in numerical results in Fig.~\ref{fig:nonstationary}.} 
For the Gaussian-QAM mixture of Fig.~\ref{fig:USRP_Gauss_qam_}, AOP-LMS is most superior since it is able to tailor its orthonormal polynomials specifically to the signal distribution; IH-LMS is not far behind, however, given its strength in handling Gaussian signals.
Overall, these experimental results demonstrate that the proposed algorithm, AOP-LMS, is the most capable in canceling real-world SI across a variety of signal distributions. 
A full video of the experimental demonstration can be accessed at \textit{https://wireless-x.korea.ac.kr/full-duplex-radios}.

\section{Conclusion}

In this paper, we introduced a new digital SIC algorithm referred to as AOP-LMS that boasts remarkable robustness in the face of non-stationary data distributions. 
The innovative concept behind the proposed AOP-LMS algorithm involves the adaptive formation of basis functions by estimating the moments of the data symbols. 
This research has uncovered that it is feasible to construct  {orthonormal} basis functions through the generalization of HPs, even for arbitrarily distributed input data. 
Our simulations confirmed that the AOP-LMS SIC algorithm outperforms current state-of-the-art SIC algorithms when dealing with non-stationary input distributions. 
Moreover, by augmenting AOP-LMS with a pre-computed LUT based on the known set of MCSs a priori can reduce the overall computational costs at run-time in practical wireless systems.
We substantiated all of these results through both extensive simulation and experimentation, solidifying its potential role in unlocking FD wireless systems for next-generation networks.
Relevant future work includes extending the proposed techniques to multi-antenna systems and jointly optimizing model parameters by leveraging advancements in machine learning.

\begin{appendices}
\section*{Appendix}
\section*{Proof of Theorem 1}
For finding coefficients of  {a} $(2p-1)$th order polynomial where $p>1$, we assume that coefficients of polynomials whose orders are lower than $p$ satisfy each order of orthogonal conditions. The $(2p-1)$th order polynomial should have coefficients that satisfy $(p-1)$ conditions as follows:
\begin{align}\label{eq:pth_iter}
    \mathbb{E}[\phi_{p}^{*}(x[n];\mathbf{c}_{p})\phi_{k}(x[n];\mathbf{c}_{k})] = 0, \quad k=1,2,\dots,p-1.
\end{align}
Rewriting \eqref{eq:pth_iter} for all $k$, we obtain
\begin{align}\label{eq:Ortho_cond_for_pth}
    \begin{cases}
        \mathbb{E}[\phi_{p}^{*}(x[n];\mathbf{c}_{p})  c_{1,0}x[n]]=0, & (k=1), \\
        \mathbb{E}[\phi_{p}^{*}(x[n];\mathbf{c}_{p})   (c_{2,1}|x[n]|^2x+c_{2,0}x[n])]=0, & (k=2), \\
        \quad \quad \vdots \\
        \mathbb{E}[\phi_{p}^{*}(x[n];\mathbf{c}_{p}) \phi_{p-1}(x[n];\mathbf{c}_{p-1})]=0,  & (k=p-1),
    \end{cases}
\end{align}
where $\mathbf{c}_{k}$  {are} determined coefficients that satisfy the orthogonal condition.
Here, we consider that polynomials are monic polynomials, i.e., $c_{p,p-1}=1, \forall p$.
To construct the orthogonal basis functions, we require the moment information $\pmb{\mu}_2^{4p-4}$ as in \eqref{eq:Ortho_cond_for_pth}.

Let the unknown parameter vector $\bar{\mathbf{c}}_{p} = [c_{p,0},c_{p,1}\dots,c_{p,p-2}]^\top$ be a subvector of  $\mathbf{c}_p=[\bar{\mathbf{c}}_{p}, 1]^\top$. Then, the orthogonality conditions in \eqref{eq:Ortho_cond_for_pth} boil down to a matrix form:
\begin{align}
    \mathbf{C}_p\mathbf{M}_p\bar{\mathbf{c}}_p + \mathbf{C}_p \pmb{\mu}_{2p}^{4p-4} &= \mathbf{0},
\end{align}
where $\mathbf{C}_p$ is a lower triangular matrix. For example,when $p=3$, the orthogonality condition in the matrix form is given by
    \begin{align}\begin{bmatrix}
    1& 0 \\  c_{2,0} & 1
    \end{bmatrix}
    \begin{bmatrix}
    \mu_2 & \mu_4 \\ \mu_4  & \mu_6
    \end{bmatrix}
    \begin{bmatrix}
    c_{3,0} \\ c_{3,1}
    \end{bmatrix}
    + \begin{bmatrix}
    1& 0 \\  c_{2,0} & 1
    \end{bmatrix}
    \begin{bmatrix}
    \mu_6 \\ \mu_8
    \end{bmatrix} = \mathbf{0}.
    \end{align}
 Since the matrix $\mathbf{C}_p$ is invertible, the orthogonality condition boils down to
 \begin{align}
 \mathbf{M}_p\bar{\mathbf{c}}_p +  \pmb{\mu}_{2p}^{4p-4} = \mathbf{0}.
 \end{align}
 Consequently, the unique coefficient vector $\bar{\mathbf{c}}_p$ exists, provided that the Hankel matrix $\mathbf{M}_p$ is invertible. 

Since we started with a monic polynomial, usually the self inner product of a polynomial is not $1$. 
Since the power of the $p$th polynomial is $\mathbb{E}[|\phi_p(x;\mathbf{c}_p)|^2]$, the normalization factor  {becomes} 
\begin{align} \label{eq:pol_normalization_factor}
    z = \sqrt{\mathbb{E}[|\phi_p(x;\mathbf{c}_p)|^2]}.
\end{align}
 This completes the proof.

\section*{Proof of Theorem 2}

The response function of the nonlinear PA is approximated with  {an} odd order polynomial \cite{Book:ding2004digital},\cite{Raich2002} as
\begin{align}\label{eq:memory_PA_response_scalar}
    f(x[n]) = \sum_{k=0}^{\frac{P-1}{2}} \sum_{m=0}^{L-1} b_{km} x[n-m]|x[n-m]|^{2k},
\end{align}
where $b_{km}$ is a coefficient of the response function and $x$ is the baseband power amplifier input. In wireless communication, the response function of the PA  {offer} only  {consider} odd order  {terms, since} the signals generated by even order terms are  {out of the band of interest}. For vector notation, let $\phi_p(x)$ be a polynomial function defined as
\begin{align}
    \phi_p(x) = x|x|^{2(p-1)},
\end{align}
 $\pmb{\phi}_p(\mathbf{x}[n]) = [\phi_p(x[n]), \phi_p(x[n-1]), \dots, \phi_p(x[n-(L-1)])]^{\sf H}$, and the weight vector $\bar{\mathbf{h}}_k[n] = [b_{k0},b_{k1},\dots, b_{k(L-1)}]^{\sf H}$.
 Then, \eqref{eq:memory_PA_response_scalar} becomes 
 \begin{align}\label{eq:Memory_vectorized_representation}
     f(x[n]) = \sum_{k=1}^{\frac{P+1}{2}} \pmb{\phi}_k(\mathbf{x}[n])^{\sf H} \bar{\mathbf{h}}_k[n].
 \end{align}

 Let $\phi_p^{\sf OP}(x)$ be a orthonormal polynomial function which satisfies condition \eqref{eq:orthogonal_condition}. From  Theorem 1, we know that  {orthonormal} polynomials are expressed as linear combinations of $\phi_p(x)$. In addition, when the Hankel matrix consisting of moments of a signal is non-singular, the coefficients constitute a full-rank lower triangle matrix. Therefore, the vectorized polynomial and vectorized orthonormal polynomial have the following relationship:
\begin{align}\label{eq:relation_vectorized_basis}
    \begin{bmatrix}
        \pmb{\phi}^{\sf OP}_1(\mathbf{x}[n])^{\sf H} \\ \pmb{\phi}^{\sf OP}_2(\mathbf{x}[n])^{\sf H}  \\ \vdots \\ \pmb{\phi}^{\sf OP}_p(\mathbf{x}[n])^{\sf H} 
    \end{bmatrix}
    = 
    \mathbf{C}_p 
    \begin{bmatrix}
        \pmb{\phi}_1(\mathbf{x}[n])^{\sf H} \\ \pmb{\phi}_2(\mathbf{x}[n])^{\sf H} \\ \vdots \\ \pmb{\phi}_p(\mathbf{x}[n])^{\sf H} 
    \end{bmatrix},
\end{align}
 where $\mathbf{C}_p$ is invertible coefficient matrix. When merging \eqref{eq:Memory_vectorized_representation} and \eqref{eq:relation_vectorized_basis}, the nonlinear response is represented by the orthonormal basis as
 \begin{align}
     f(x[n]) = \sum_{k=1}^{\frac{P+1}{2}} \pmb{\phi}_k^{\sf OP}(\mathbf{x}[n])^{\sf H} \mathbf{h}_k[n],
 \end{align}
where 
 \begin{align}
     \begin{bmatrix}
         \mathbf{h}_1[n] & \dots &\mathbf{h}_p[n]
     \end{bmatrix}
     =
     \begin{bmatrix}
         \bar{\mathbf{h}}_1[n] & \dots & \bar{\mathbf{h}}_p[n]
     \end{bmatrix}
     \mathbf{C}_p^{-1}.
 \end{align}
This completes the proof.

\end{appendices}

\bibliographystyle{IEEEtran}
\bibliography{IEEEabrv,Journal_bib}

\end{document}